\documentclass[
aps,%
pre,%
reprint,%
showpacs,%
superscriptaddress,%
groupedaddress,%
floatfix,%
nofootinbib%
]{revtex4-1}
\usepackage{float}
\usepackage{verbatim}
\usepackage{amsmath,amsfonts,amsthm,amssymb,times,mathrsfs,txfonts}
\usepackage[hidelinks]{hyperref}
\usepackage{bm}% bold math
\usepackage{multirow}
\usepackage{enumitem}
\newcommand{\field}[1]{\mathbb{#1}}
\newcommand{\fs}[1]{\mathsf{#1}}

\DeclareMathOperator*{\supp}{supp}

\DeclareMathOperator{\Ker}{Ker}
\DeclareMathOperator{\Ran}{Ran}
\newcommand{\tp}{\intercal}% transpose operation

\newcommand{\ovl}[1]{\overline{#1}}
\newcommand{\bigO}[1]{\mathop{\mathscr{O}}(#1)}

\let\Re\relax
\DeclareMathOperator{\Re}{Re}
\let\Im\relax
\DeclareMathOperator{\Im}{Im}
\newcommand{\vv}[1]{\boldsymbol{#1}}
\newcommand{\vs}[1]{\boldsymbol{#1}}

\newcommand{\OP}[1]{\mathscr{#1}}

\DeclareMathOperator{\fourier}{\mathscr{F}}
\DeclareMathOperator{\Si}{Si}
\DeclareMathOperator{\Cin}{Cin}
\DeclareMathOperator{\Ci}{Ci}
\DeclareMathOperator{\si}{si}
\DeclareMathOperator{\sinc}{sinc}
\DeclareMathOperator{\besselj}{j}
\DeclareMathOperator{\legendre}{L}

\newtheorem{theorem}{Theorem}[section]
\newtheorem{prop}[theorem]{Proposition}
\newtheorem{corr}[theorem]{Corollary}
\newtheorem{lemma}[theorem]{Lemma}

\newtheorem{rem}{Remark}[section]
\renewcommand{\Re}{\text{Re}}
\renewcommand{\Im}{\text{Im}}

\newcommand{\wtilde}[1]{\widetilde{#1}}

\usepackage[pdftex]{graphicx}
\graphicspath{{./figs/}}
\usepackage[caption=false]{subfig}

\begin{document}

%\preprint{}

\title{Nonlinearly Bandlimited Signals}


\author{Vishal Vaibhav}
\email[]{vishal.vaibhav@gmail.com}
%\affiliation{Delft Center for Systems and Control, 
%Delft University of Technology, Mekelweg 2. 2628 CD Delft, 
%The Netherlands}% <-this % stops a space
\noaffiliation{}
\date{\today}

\begin{abstract}
In this paper, we study the inverse scattering problem for a class of signals
that have a compactly supported reflection coefficient. The problem boils down
to the solution of the Gelfand-Levitan-Marchenko (GLM) integral equations with a
kernel that is bandlimited. By adopting a sampling theory approach
to the associated Hankel operators in the Bernstein spaces, a constructive proof
of existence of a solution of the GLM equations is obtained under various
restrictions on the nonlinear impulse response (NIR). The formalism developed in this
article also lends itself well to numerical computations yielding algorithms
that are shown to have algebraic rates of convergence. In particular, the use Whittaker-Kotelnikov-Shannon 
sampling series yields an algorithm that converges as $\bigO{N^{-1/2}}$ whereas the use
of Helms and Thomas (HT) version of the sampling expansion yields an algorithm that
converges as $\bigO{N^{-m-1/2}}$ for any $m>0$ provided the regularity
conditions are fulfilled. The complexity of the algorithms
depend on the linear solver used. The use of conjugate-gradient (CG) method yields an
algorithm of complexity $\bigO{N_{\text{iter.}}N^2}$ per sample of the signal 
where $N$ is the number of sampling basis functions used and $N_{\text{iter.}}$
is the number of CG iterations involved. The HT version of the sampling
expansions facilitates the development of algorithms of complexity 
$\bigO{N_{\text{iter.}}N\log N}$ (per sample of the signal) by exploiting the special structure as well as the 
(approximate) sparsity of the matrices involved. The algorithms are numerically 
validated using Schwartz class functions as NIRs that are either bandlimited or
effectively bandlimited. The results suggest that the HT variant of our
algorithm is spectrally convergent for an input of the aforementioned class.
\end{abstract}

%\keywords{}

\maketitle

%\section*{Notations}
%%%%%%%%%%%%%%%%%%%%%%%%%%%%%%%%%%%%%%%%%%%%%%%%%%%%%%%

\section{Introduction}
In this paper, we address the inverse scattering problem for a class of signals
such that the continuous part of their nonlinear Fourier spectrum has a 
compact support and the discrete part is empty. Such signals are called 
\emph{nonlinearly bandlimited} in analogy with bandlimited signals in 
conventional Fourier analysis and they are entirely radiative in nature 
by definition. Such signals lend themselves well to the design of a fast 
inverse nonlinear Fourier transform algorithm~\cite{V2018BL,V2017INFT1} in the differential approach
of inverse scattering. In practical applications, such signals are often 
used as a convenient approximation of signals that have an effectively localized
continuous spectrum. This problem, for instance in the Hermitian class, arises in the design of 
nonuniform fiber Bragg gratings to compensate for second and third order
dispersion in optical fibers~\cite{SSFSS1998,FZM1999}. The target reflection coefficient 
in these problems is a compactly supported chirped profile. 
In the non-Hermitian class, the design of grating-assisted co-directional couplers, a 
device used to couple light between two different guided modes of an optical fiber 
(see~\cite{FZ2000,BS2003} and references therein) requires the solution of a
similar problem. Such signals have
also attracted interest in optical communication where it is
proposed to encode information in the continuous part of nonlinear Fourier
spectrum in an attempt to mitigate nonlinear signal distortions at higher power
levels~\cite{TPLWFK2017}.

In all of the applications mentioned above, accuracy of the numerical
algorithms form a bottleneck either at higher powers in the non-Hermitian class or
at reflectivities approaching unity in the Hermitian class. There is a vast amount
literature on numerical methods for the solution of the Gelfand-Levitan-Marchenko 
(GLM) integral equations notable among them are the integral
layer-peeling~\cite{RH2003}, T\"oplitz
inner-bordering~\cite{BFPS2007,FBPS2015,V2018TL} and 
the Nystr\"om method~\cite{MST2007}. From a practical viewpoint, these
algorithms work for a large class of problems; however, these methods cannot
provide accuracies upto the machine precision with the exception of the method due to
Trogdon and Olver~\cite{TO2013}. This method relies on the formulation of the 
inverse scattering problem as a Rieman-Hilbert problem and it has been demonstrated to be
spectrally convergent. Its domain of application is not limited to the class
of signals considered in this article; however, the complexity of this
algorithm remains high at the same time it is somewhat complicated to implement.

The inverse scattering problem is 
generally formulated on an unbounded domain which poses a serious problem for the
underlying quadrature schemes in the Nystr\"om method or for the overlap integrals 
in the degenerate Kernel method (see Atkinson~\cite{A2009} for an introduction to
these methods). In this paper, following Khare and George~\cite{KG2003} (see
also Vaibhav~\cite{V2018PSWFs}), we propose a sampling theory based approach to the 
discretization of the GLM equations which has the advantage that the basis functions 
are naturally adapted to unbounded domains. The bandlimited nature of the 
functions facilitate accurate quadrature on unbounded domains~\cite{M2001}. It is 
important to emphasize that the method thus obtained requires sampling of the 
impulse response on an equispaced grid which has some clear advantages in 
preserving the inherent symmetries of the system. 

In the sampling approach presented in this paper, we use the classical 
Whittaker-Kotelnikov-Shannon sampling series and the Helms and Thomas (HT) 
version~\cite{HT1962,J1966} of the sampling expansion. The associated basis functions
provide a natural framework for the representation of the Hankel operators 
involved which makes the theoretical analysis related to existence of 
solution or issues of convergence somewhat easier. Further, the Bernstein 
spaces~\cite{M2001} provide a natural setting for rigorous analysis 
of the GLM equations. 

The algorithms presented in this article are shown to have
algebraic orders of convergence. In particular, the use the HT
version of the sampling expansion affords an accuracy of $\bigO{N^{-m-1/2}}$ 
(provided certain regularity conditions are fulfilled) where $N$ is the number 
of basis function used and $m>0$ is a parameter that 
can be chosen arbitrarily. The complexity of these algorithms depends on the
linear solver used. In order to compute one sample of the signal with a direct 
solver, the algorithm would require $\bigO{N^3}$ operations whereas an iterative
solver based on the conjugate-gradient method yields the same result in
$\bigO{N_{\text{iter.}}N^2}$ operations. It must be emphasized that at any step, good seed 
solutions are readily available (from the previous step) when the signal is
being computed on a sufficently fine grid so that quantity $N_{\text{iter.}}$ does
not become prohibitively large. Further, the HT version of the sampling 
series leads to a dramatic decrease in complexity within the iterative approach 
if one takes into account the special
structure as well as the (approximate) sparsity of the matrices involved. The
sparsity structure can be controlled by introducing a tolerance $\epsilon>0$
which introduces an error of $\bigO{N\epsilon}$ while reducing 
the complexity to $\bigO{N_{\text{iter.}}(\epsilon)N\log N}$ per sample of the signal.

The rest of the paper is organized as follows: Sec.~\ref{sec:GLM-BL} discusses
the GLM equations in the functional spaces introduced in Sec.~\ref{sec:prelims}.
The exposition is organized such that the properties of the Hankel operators is
studied in Sec.~\ref{sec:Hnakel-BL} which is then used to discuss the GLM
equation in Sec.~\ref{sec:sampling-approach}. Sec.~\ref{sec:HT} discusses
the application of the HT version of the sampling expansion. Sec.~\ref{sec:num} 
deals with the numerical and algorithmic aspects of the ideas
developed in the preceding section. Sec

\section{Gelfand-Levitan-Marchenko Equations with Bandlimited Kernels}
\label{sec:GLM-BL}
The coupled Gelfand-Levitan-Marchenko (GLM) integral equations arise in 
connection with the inverse scattering problem for the Hermitian as well as the 
non-Hermitian $2\times2$ Zakharov-Shabat scattering 
problem~\cite{ZS1972,AKNS1974,AS1981}. As stated earlier, we consider a class of
signals such that its reflection coefficient $\rho(\xi)\,(\xi\in\field{R})$ has
a compact support, say, in $[-\sigma,\sigma]$ where $\sigma$ is referred to as the
\emph{bandlimiting parameter}. The nonlinear impulse response (NIR), defined by
\begin{equation}
p(\tau)=\frac{1}{2\pi}\int_{-\sigma}^{\sigma}\rho(\xi) e^{i\xi\tau}d\xi,
\end{equation}
is, evidently, a bandlimited function. 

Let $q(t)$ denote the inverse nonlinear Fourier transform (NFT) of $\rho(\xi)$.
The GLM equations corresponding to $p(\tau)$ can be stated as
\begin{equation}\label{eq;GLM-start}
\begin{split}
&\kappa {A}_2^*(\tau,t)=\int_{\tau}^{\infty}{A}_1(\tau,s){p}(s+t)ds,\\
&{A}_1^*(\tau,t) = p(\tau+t) +\int_{\tau}^{\infty}A_2(\tau,s){p}(s+t)ds,
\end{split}
\end{equation}
where $\kappa=+1$ (Hermitian scattering problem) and $\kappa=-1$ (non-Hermitian
scattering problem). The scattering potential is recovered from 
\begin{equation}
q(\tau)=-2A_1(\tau,\tau),
\end{equation}
together with the estimate
\begin{equation}
\int^{\infty}_{\tau}|q(s)|^2ds=2\kappa A_2(\tau,\tau).
\end{equation}
Let us enumerate two interesting properties that will be
useful later:
\begin{itemize}
\item \textit{Shift in time domain}: If $q(t)$ is the 
inverse NFT of $\rho(\xi)$, then the inverse NFT
of $\rho(\xi)e^{2i\xi t_0}$ is $q(t+t_0)$. 
\item\textit{Scaling in frequency domain}: If the inverse NFT of $\rho(\xi)$ is $q(t)$, then the inverse
NFT of $\rho(\lambda\xi)$ is $\lambda^{-1} q(t/\lambda)$.
\end{itemize}
The shifting property allows us to fix $\tau=\tau_0$ in~\eqref{eq;GLM-start} and 
simply keep varying the variable $t_0$ to obtain the scattering potential over 
the entire real line. Therefore, we may set $\tau=0$ in~\eqref{eq;GLM-start} without 
the loss of generality and focus on the following form of the GLM equations:
\begin{equation}\label{eq;GLM-mod}
\begin{split}
&\kappa {A}_2^*(t)=\int_{0}^{\infty}{p}(s+t){A}_1(s)ds,\\
&{A}_1^*(t) = p(t) +\int_{0}^{\infty}{p}(s+t)A_2(s)ds,
\end{split}
\end{equation}
%%%%%%%%%%%%%%%%%%%%%%%%%%%%%%%%%%%%%%%%%%%%
%%%%%%%%%%%%%%%%%%%%%%%%%%%%%%%%%%%%%%%%%%%%
%%%%%%%%%%%%%%%%%%%%%%%%%%%%%%%%%%%%%%%%%%%%
\subsection{Preliminaries}\label{sec:prelims}
The set of real numbers (integers) is denoted by $\field{R}$ ($\field{Z}$) and 
the set of non-zero positive real numbers (integers) by $\field{R}_+$ 
($\field{Z}_+$). The set of complex numbers are denoted by $\field{C}$,
and, for $\zeta\in\field{C}$, $\Re(\zeta)$ and $\Im(\zeta)$ refer to the real
and the imaginary parts of $\zeta$, respectively. The complex conjugate of 
$\zeta\in\field{C}$ is denoted by $\zeta^*$ and $\sqrt{\zeta}$ denotes its
square root with a positive real part. The upper-half (lower-half) of $\field{C}$ 
is denoted by $\field{C}_+$ ($\field{C}_-$) and it closure by $\ovl{\field{C}}_+$
($\ovl{\field{C}}_-$).

The Fourier transform of a function $f(t)$ is defined as
\begin{equation*}
F(\xi)=\fourier[f](\xi)=
\int_{\field{R}}f(t)e^{-i\xi t}dt.
\end{equation*}
The characteristic function of a set $\Omega\subset\field{R}$ is denoted by
\begin{equation}
\chi_{\Omega}=
\begin{cases}
1,& x\in\Omega,\\
0,& \text{otherwise}.
\end{cases}
\end{equation}
The Lebesgue spaces over the domain $\Omega\subset\field{R}$ are denoted by 
$\fs{L}^{\nu}(\Omega)\,(1\leq\nu\leq\infty)$ and corresponding norm by
$\|\cdot\|_{\fs{L}^{\nu}(\Omega)}$. If the domain is not mentioned, it assumed to
be $\field{R}$ unless otherwise stated. 

For a rigorous analysis of the GLM equations 
with bandlimited kernels, it is
convenient to work with the \emph{Bernstein spaces}~\cite[Chap.~2]{M2001} (also
see~\cite[Chap.~3]{N1975}), $\fs{B}_{\sigma}^{\nu}$ with 
$1\leq\nu\leq\infty$, defined as the class of entire
functions of exponential type-$\sigma$ whose restriction to the $\field{R}$ belong
to $\fs{L}^{\nu}$. Further, these spaces satisfy the following embedding property:
$\fs{B}_{\sigma}^{\nu}\subset \fs{B}_{\sigma}^{\nu'}\subset \fs{B}_{\sigma}^{\infty}$ where
$1\leq\nu\leq\nu'\leq\infty$. For $f\in\fs{B}_{\sigma}^{\nu}$ and any $h>0$, the following
inequality holds
\begin{equation}\label{eq:B-spaces-norm}
\|f\|_{\fs{L}^{\nu}}\leq\sup_{t\in\field{R}}
\left(\sum_{n\in\field{Z}}\left|f\left(t-nh\right)\right|^{\nu}\right)^{1/\nu}
\leq (1+\sigma h)\|f\|_{\fs{L}^{\nu}}.
\end{equation}
This inequality proves extremely useful in establishing certain bounds and it
appears mostly with the parameter $h=\pi/\sigma$, the grid
spacing for \emph{Nyquist sampling} of $\sigma$-bandlimited functions. Further,
we recall from Boas~\cite[Thm.~6.7.1]{Boas1954}, if 
$f\in\fs{B}^{\nu}_{\sigma}\, (1\leq\nu<\infty)$, then
\begin{equation}
\int^{\infty}_{-\infty}|p(x+iy)|^{\nu}dx\leq 
e^{\nu\sigma|y|}\int^{\infty}_{-\infty}|p(x)|^{\nu}dx,
\end{equation}
and $\lim_{|t|\rightarrow\infty}f(t)=0$. 

Next let us define the \emph{Hardy classes} $\fs{H}^2_{\pm}$ 
which are a class of functions analytic in upper (lower) half of the 
complex plane such that the expressions (which qualify as norms)
\begin{equation}
\begin{split}
&\|f\|_{\fs{H}^2_+}=\sup_{\eta\in\field{R}_+}\left(\int_{\field{R}}f(\xi+i\eta)|^2d\xi\right)^{1/2},\\
&\|f\|_{\fs{H}^2_-}=\sup_{\eta\in\field{R}_-}\left(\int_{\field{R}}f(\xi+i\eta)|^2d\xi\right)^{1/2}.
\end{split}
\end{equation}
are bounded, respectively. The Paley-Wiener theorem allows one to characterize
these spaces solely in terms of their boundary functions as follows:
\begin{equation}
\begin{split}
&\fs{H}^2_+=\{f\in\fs{L}^2|\,\fourier^{-1}[f]|_{\Omega_+}=0\},\\
&\fs{H}^2_-=\{f\in\fs{L}^2|\,\fourier^{-1}[f]|_{\Omega_-}=0\},\\
\end{split}
\end{equation}
so that $\fs{L}^2=\fs{H}^2_+\oplus\fs{H}^2_-$. For $f\in\fs{L}^2$, the
decomposition into $\fs{H}^2_{\pm}$ reads as
\begin{equation}
\begin{split}
f^{(+)}&=\left(\fourier\circ\chi_{\Omega_-}\circ\fourier^{-1}\right)f=\frac{1}{2}\left(f+i\OP{H}[f]\right),\\
f^{(-)}&=\left(\fourier\circ\chi_{\Omega_+}\circ\fourier^{-1}\right)f=\frac{1}{2}\left(f-i\OP{H}[f]\right),
\end{split}
\end{equation}
respectively.
\begin{lemma}
If $p\in\fs{B}_{\sigma}^2$, then $\rho\in\fs{L}^1\cap\fs{L}^2$ with support in
$[-\sigma,\sigma]$.
\end{lemma}
\begin{proof}
If $p\in\fs{B}_{\sigma}^2$, then $\rho\in\fs{L}^2$ with support in
$[-\sigma,\sigma]$. Then, using Cauchy-Schwartz inequality, we have
\begin{equation}
\left(\int_{-\sigma}^{\sigma}|\rho(\xi)|d\xi\right)
\leq\sqrt{2\sigma}\int_{-\sigma}^{\sigma}|\rho(\xi)|^2d\xi
\leq\sqrt{2\sigma}\|p\|_{\fs{L}^2}.
\end{equation}
\end{proof}
The functions in $\fs{B}^{\infty}_{\sigma}$ can be regarded as Fourier-Laplace
transforms of certain class of distributions supported in $[-\sigma,\sigma]$. If 
$\rho(\xi)$ is a function of bounded variation on $(-\sigma,\sigma)$, denoted by
$\fs{BV}(-\sigma,\sigma)$, such that
$\rho(-\sigma+0)=\rho(\sigma-0)$, then $p(z)$ satisfies the
following estimate
\begin{equation}
|p(z)|\leq\frac{C}{1+|z|}e^{\sigma |\Im(z)|},\quad z\in\field{C},
\end{equation}
for some $C>0$. Such function belong $\fs{B}^2_{\sigma}$ but not
$\fs{B}^1_{\sigma}$. Further, if $\rho\in\fs{C}^{\nu}_0(\field{R})$ with support in
$[-\sigma,\sigma]$, then, there exists a $C_{\nu}>0$ such
that~\cite{Y1968}
\begin{equation}
|p(z)|\leq\frac{C}{(1+|z|)^{\nu}}e^{\sigma |\Im(z)|},\quad z\in\field{C}.
\end{equation}

\subsection{Hankel Operators with Bandlimited Kernels}\label{sec:Hnakel-BL}
Let $\Omega_+=[0,\infty)$ and define the Hankel operator
\begin{equation}\label{eq:p-GLM-OP}
\OP{P}[g](t)=\int_{\Omega_+} p(t+s)g(s)ds,\quad t\in\Omega_+.
\end{equation}
The field underlying the image of $\OP{P}$ can be extended to the entire complex
plane. Let $\Omega_-=(-\infty,0]$. For convenience, we may 
also work with the form $\wtilde{\OP{P}}[g](t)=\OP{P}[g](-t)$ so that, 
for $g$ supported in $\Omega_+$, 
\begin{equation}\label{eq:hankel-conv}
\wtilde{\OP{P}}[g](t)=\int_{\Omega_+} \tilde{p}(t-s)g(s)ds
=(\tilde{p}\star g)(t),\quad t\in\Omega_-,
\end{equation}
where $\tilde{p}(t)=p(-t)$ and ``$\star$'' denotes convolution. In the Fourier domain, 
the Hankel operator $\OP{P}$ can be expressed as
\begin{equation}
\OP{H}_{\tilde{\rho}}=\fourier\circ\chi_{\Omega_-}\circ\wtilde{\OP{P}}\circ\fourier^{-1}
\end{equation}
so that
\begin{equation}\label{eq:henkel-op-2}
\OP{H}_{\tilde{\rho}}[G](\xi)
=\left(\fourier\circ\chi_{\Omega_-}\circ\fourier^{-1}\right)[\tilde{\rho}G](\xi),\quad\xi\in\field{R},
\end{equation}
where $G(\xi)=\fourier[g](\xi)$ with $g$ supported in $\Omega_+$.

%%%%%%%%%%%%%%%%%%%%%%%%%%%%%%%%%%%%%%%%%%%%%%%%%%%%%%%%%%%%%%%%%%%%%%%%%%%%%%%%%%%%%%%%%
\begin{prop}[Boundedness of Hankel operators]\label{lemma:bounded-hankel}
Define $\mathcal{I}_{\nu}=\|p\chi_{\Omega_+}\|_{\fs{L}^{\nu}}$.
\begin{enumerate}[label=(\alph*)]
\item\label{lemma:bound-a} 
If $p\in\fs{L}^1$, then~\eqref{eq:p-GLM-OP} defines a bounded linear operator
$\OP{P}:\fs{L}^{\nu}(\Omega_+)\rightarrow\fs{L}^{\nu}(\Omega_+)$ for $\nu=1,2$.

\item\label{lemma:bound-b} 
If $p\in\fs{B}^1_{\sigma}$, then~\eqref{eq:p-GLM-OP} defines a bounded linear operator
$\OP{P}:\fs{L}^{\nu}(\Omega_+)\rightarrow\fs{B}^{\nu}_{\sigma}$ for $\nu=1,2$. 

\item\label{lemma:bound-c} 
If $p\in\fs{B}^2_{\sigma}$ with $\rho\in\fs{L}^{\infty}$,
then~\eqref{eq:p-GLM-OP} and~\eqref{eq:henkel-op-2} 
define bounded linear operators
$\OP{P}:\fs{L}^{2}(\Omega_+)\rightarrow\fs{B}^{2}_{\sigma}$ and 
$\OP{H}_{\tilde{\rho}}:\fs{H}^{2}_-\rightarrow\fs{H}^{2}_+$, respectively.
\end{enumerate}
\end{prop}
\begin{proof}
In order to prove~\ref{lemma:bound-a}, consider $g\in\fs{L}^1(\Omega_+)$, then
\begin{equation*}
\begin{split}
\|\OP{P}[g]\|_{\fs{L}^{1}(\Omega_{+})}
&\leq\int_{0}^{\infty}
\left|\int^{\infty}_{0}p(y+s)g(s)ds\right|dy,\\
&\leq\int_{0}^{\infty}\left(\int^{\infty}_{0}|p(y+s)|dy\right)|g(s)|ds,\\
&\leq\mathcal{I}_{1}\int_{0}^{\infty}|g(s)|ds
=\mathcal{I}_{1}\|g\|_{\fs{L}^{1}(\Omega_{+})}.
\end{split}
\end{equation*}
Next, let $g\in\fs{L}^2(\Omega_{+})$, then
\begin{equation*}
\begin{split}
&\|\OP{P}[g]\|^2_{\fs{L}^{2}(\Omega_{+})}
\leq\int_{0}^{\infty}
\left|\int^{\infty}_{0}p(y+s)g(s)ds\right|^2dy,\\
&\leq\int_{0}^{\infty}dy\left[
\int^{\infty}_{0}|p(y+s)|ds\int^{\infty}_{0}|p(y+s)||g(s)|^2ds\right],\\
&\leq\mathcal{I}_{1}\int_{0}^{\infty}
\left(\int^{\infty}_{0}|p(y+s)|dy\right)|g(s)|^2ds
\leq\mathcal{I}^2_{1}\|g\|^2_{\fs{L}^{2}(\Omega_{+})}.
\end{split}
\end{equation*}
Therefore, we have
\begin{equation}
\|\OP{P}[g]\|_{\fs{L}^{\nu}(\Omega_{+})}
\leq\mathcal{I}_{1}\|g\|_{\fs{L}^{\nu}(\Omega_{+})},\quad\nu=1,2.
\end{equation}
This completes the proof of statement~\ref{lemma:bound-a}. Note that it is also 
possible to show that
\begin{equation}\label{eq:L1-estimate}
\|\OP{P}[g]\|_{\fs{L}^{\nu}}
\leq\|p\|_{\fs{L}^{\nu}}\|g\|_{\fs{L}^{\nu}(\Omega_{+})},\quad\nu=1,2.
\end{equation}

To prove~\ref{lemma:bound-b}, let $z=x+iy\in\field{C}$; then, the analyticity 
property of $\OP{P}[g](z)$ follows from the analyticity of $p(z)$ and the 
Bernstein's inequality~\cite[Chap.~3]{N1975} which ensures that its derivative
$p'\in\fs{B}^1_{\sigma}$. The boundedness of
$\OP{P}:\fs{L}^{\nu}(\Omega_+)\rightarrow\fs{L}^{\nu}$ follows
from~\eqref{eq:L1-estimate}. What remains to show is that $\OP{P}[g](z)$ is of 
exponential type-$\sigma$. Observing
\begin{equation}
|\OP{P}[g](z)|\leq Ce^{\sigma|y|}\int^{\infty}_{0}|g(s)|ds,
\end{equation}
the result follows for $\nu=1$. For $\nu=2$, we have
\begin{equation*}
\begin{split}
|\OP{P}[g](z)|^2
&\leq\int_{0}^{\infty}
\left|\int^{\infty}_{0}p(z+s)g(s)ds\right|^2dy,\\
&\leq\int^{\infty}_{0}|p(z+s)|^2ds\int^{\infty}_{0}|g(s)|^2ds.
\end{split}
\end{equation*}
From Boas~\cite[Thm.~6.7.1]{Boas1954} and using the fact that
$\fs{B}^1_{\sigma}\subset\fs{B}^2_{\sigma}$, we have 
\begin{equation*}
\int^{\infty}_{0}|p(z+s)|^2ds\leq 
e^{2\sigma|y|}\int^{\infty}_{0}|p(x+s)|^2ds,
\end{equation*}
which yields the estimate
\begin{equation}
|\OP{P}[g](z)|\leq 
Ce^{\sigma|y|}\|g\|_{\fs{L}^{2}(\Omega_{+})},
\end{equation}
for some $C>0$. Finally, for any $g\in\fs{L}^{1}$ and $h>0$,
\begin{equation*}
\begin{split}
\sup_{t\in\field{R}}
\left(\sum_{n\in\field{Z}}\left|\OP{P}[g]\left(t-nh\right)\right|\right)
&\leq\sup_{t\in\field{R}}
\left(\sum_{n\in\field{Z}}\left|p\left(t-nh\right)\right|\int_0^{\infty}|g(s)|ds\right)\\
&\leq (1+\sigma h)\|p\|_{\fs{L}^1}\|g\|_{\fs{L}^1(\Omega_+)}.
\end{split}
\end{equation*}

To prove~\ref{lemma:bound-c}, first we consider $\OP{H}_{\tilde{\rho}}$. 
It is an elementary exercise to verify that, for any $G\in\fs{H}_-^2$,
\begin{equation}
\|\OP{H}_{\tilde{\rho}}[G]\|_{\fs{L}^2}
\leq\|\tilde{\rho}\|_{\fs{L}^{\infty}}\|G\|_{\fs{L}^2}.
\end{equation}
Turning to $\OP{P}$, if $p\in\fs{B}^2_{\sigma}$, then it is straightforward to
see that $\OP{P}[g]\in\fs{B}^{\infty}_{\sigma}$ for any
$g\in\fs{L}^2$ with support in $\Omega_+$. Using Plancheral's theorem, we have
\begin{equation}
\|\OP{P}[g]\|_{\fs{L}^2}
\leq(2\pi)^{-1/2}\|\tilde{\rho}G\|_{\fs{L}^2}
\leq \|\tilde{\rho}\|_{\fs{L}^{\infty}}\|g\|_{\fs{L}^2(\Omega_+)}.
\end{equation}
\end{proof}
%%%%%%%%%%%%%%%%%%%%%%%%%%%%
In the following, we would like to develop a representation of the
Hankel operators with bandlimited kernels by exploiting the sampling expansion 
of such functions. The translates of the sinc function form an 
orthonormal basis in $\fs{B}^2_{\sigma}$. Let us introduced the normalized 
form of these basis function as
\begin{equation}
\psi_n(t) = \sqrt{\frac{\sigma}{\pi}}\sinc[\sigma(t-t_n)]
=\sqrt{\frac{\sigma}{\pi}}\frac{\sin(\sigma t-n\pi)}{(\sigma t-n\pi)},
\end{equation}
so that the orthonormality condition can be stated as
\begin{equation}
\int_{-\infty}^{\infty}\psi_m(t)\psi_n(t)dt=\delta_{mn}.
\end{equation}
For $p\in\fs{B}^1_{\sigma}$, we can write 
\begin{equation}\label{eq:S-series}
p(t+s)=\sum_{n\in\field{Z}}\sqrt{\frac{\pi}{\sigma}}
p\left(t+\frac{n\pi}{\sigma}\right)\psi_n(s).
\end{equation}
This series converges absolutely and uniformly with respect to $s\in\field{R}$
where we have fixed $t\in\field{R}$. Using this representation
in~\eqref{eq:p-GLM-OP}, we have
\begin{equation}\label{eq:OP-P-series}
\begin{split}
\OP{P}[g](t)
&=\sum_{n\in\field{Z}}\sqrt{\frac{\pi}{\sigma}}
p\left(t+\frac{n\pi}{\sigma}\right)\int^{\infty}_{0}g(s)\psi_n(s)ds\\
&\equiv\sum_{n\in\field{Z}}\sqrt{\frac{\pi}{\sigma}}
p\left(t+\frac{n\pi}{\sigma}\right)\hat{g}_n.
\end{split}
\end{equation}
In view of the expression above, we introduce the Hankel operator 
\begin{equation}\label{eq:psi-OP}
\OP{S}[g](y)=\int^{\infty}_{0}\psi_0(y+s)g(s)ds=\hat{g}(y).
\end{equation}
then $\hat{g}_n=\hat{g}(-n\pi/\sigma)$. Note that $\psi_0\in\fs{B}_{\sigma}^2$
with its Fourier transform $\sqrt{{\pi}/{\sigma}}\,\chi_{[-\sigma,\sigma]}\in\fs{L}^{\infty}$, therefore,
it is a bounded linear operator from $\fs{L}^2(\Omega_+)$ to
$\fs{B}_{\sigma}^2$ (see Lemma~\ref{lemma:bounded-hankel}). Also, we have
\begin{equation}\label{eq:OP-S-norm}
\|\OP{S}\|_{\fs{L}^2}\leq\sqrt{{\pi}/{\sigma}}\|\chi_{[-\sigma,\sigma]}\|_{\fs{L}^{\infty}}
\leq\sqrt{{\pi}/{\sigma}}.
\end{equation}
In the following, we 
assume that $g\in\fs{B}^{2}_{\sigma}$. Let us show that the 
series on the right hand side of~\eqref{eq:OP-P-series} 
converges absolutely and uniformly for $t\in\field{R}$. Observing 
that $|\hat{g}_n|\leq\|g\|_{\fs{L}^{2}}$, we have
\begin{equation}
\begin{split}
\sum_{n\in\field{Z}}\sqrt{\frac{\pi}{\sigma}}
\left|p\left(t+\frac{n\pi}{\sigma}\right)\hat{g}_n\right|
&\leq\|g\|_{\fs{L}^{2}}\sum_{n\in\field{Z}}
\sqrt{\frac{\pi}{\sigma}}\left|p\left(t+\frac{n\pi}{\sigma}\right)\right|\\
&\leq (1+\pi)\|g\|_{\fs{L}^{2}}\|p\|_{\fs{L}^{1}},
\end{split}
\end{equation}
which follows from~\eqref{eq:B-spaces-norm} and the fact that $p\in\fs{B}^1_{\sigma}$. 
Further, the truncated series
\begin{equation}\label{eq:OP-P-finite}
\begin{split}
\OP{P}_N[g](t)
&=\sum_{|n|\leq N}\sqrt{\frac{\pi}{\sigma}}
p\left(t+\frac{n\pi}{\sigma}\right)\hat{g}_n,\quad N\in\field{R}_+,
\end{split}
\end{equation}
converges in $\fs{L}^{\nu}$-norm ($\nu=1,2$). To see this, consider
\begin{equation}\label{eq:OP-p-convg}
\begin{split}
\|\OP{P}_N[g]-\OP{P}_{N-1}[g]\|_{\fs{L}^{\nu}}
\leq\sqrt{{\pi}/{\sigma}}\|p\|_{\fs{L}^{\nu}}(|\hat{g}_{N}|+\hat{g}_{-N}|).
\end{split}
\end{equation}
On account of $\hat{g}\in\fs{B}^2_{\sigma}$, 
\[
\lim_{N\rightarrow\infty}|\hat{g}_{\pm N}|=0.
\]
In particular, we have $|\hat{g}_{-N}|\leq (1/\pi)\|g\|_{\fs{L}^{2}}N^{-1/2}$.
The expression on right hand side in~\eqref{eq:OP-p-convg} goes to $0$ as 
$N\rightarrow\infty$; therefore, $\OP{P}_N[g]$ is a Cauchy sequence in 
$\fs{L}^{\nu}$ and the limit belongs to $\fs{B}^1_{\sigma}$.

The sampling series in~\eqref{eq:S-series} also holds for
$p\in\fs{B}^2_{\sigma}$. In order to prove the absolute and uniform convergence of the
series in~\eqref{eq:OP-P-series} with respect to $t\in\field{R}$, we first prove that
$(\hat{g}_n)_{n\in\field{Z}}$ belongs to $\ell^2$. On account of 
$\hat{g}\in\fs{B}^2_{\sigma}$, we have
\begin{equation*}
\|\hat{\vv{g}}\|_{\ell^2}=\left(\sum_{n\in\field{Z}}|\hat{g}_n|^2\right)^{1/2}
\leq(1+\pi)\|\hat{g}\|_{\fs{L}^2}
\leq(1+\pi)\sqrt{\pi/\sigma}\|g\|_{\fs{L}^2},
\end{equation*}
which follows from~\eqref{eq:B-spaces-norm} and~\eqref{eq:OP-S-norm} so that
\begin{equation*}
\begin{split}
\sqrt{\frac{\pi}{\sigma}}\left(\sum_{n\in\field{Z}}
\left|p\left(t+\frac{n\pi}{\sigma}\right)\hat{g}_n\right|^2\right)^{\frac{1}{2}}
&\leq\sqrt{\frac{\pi}{\sigma}}\left(\sum_{n\in\field{Z}}
\left|p\left(t+\frac{n\pi}{\sigma}\right)\right|^2\right)^{\frac{1}{2}}
\|\hat{\vv{g}}\|_{\ell^{2}}\\
&\leq(1+\pi)^2(\pi/\sigma)\|g\|_{\fs{L}^2}\|p\|_{\fs{L}^2}.
\end{split}
\end{equation*}
The convergence of $\OP{P}_N[g]$ in
${\fs{L}^{2}(\Omega_{+})}$ follows the same line of reasoning as that of the 
previous case. The discussion so far can be summarized in the following
proposition:
\begin{prop}
For $p\in\fs{B}^{2}_{\sigma}$, the partial sums defined in~\eqref{eq:OP-P-finite} 
converge absolutely and uniformly (with respect to $t\in\field{R}$) for every $g\in\fs{B}^{2}_{\sigma}$ as 
$N\rightarrow\infty$. Moreover, the partial sums also converge in the 
$\fs{L}^2$-norm. If $p\in\fs{B}^{1}_{\sigma}$, then the partial sums converge 
in the $\fs{L}^1$-norm. 
\end{prop}

Next, we would like to address the problem of estimating the truncation error
which is given by
\begin{equation}\label{eq:trunc-OP-P}
\begin{split}
T_N(t)
&=\sum_{|n|>N}\sqrt{\frac{\pi}{\sigma}}
p\left(t+\frac{n\pi}{\sigma}\right)\hat{g}_n,\quad N\in\field{R}_+,
\end{split}
\end{equation}
under a stronger decay condition on $p$.
\begin{prop}\label{prop:jagerman-type}
For $p$ satisfying an estimate of the form 
\begin{equation}
|p(z)|\leq\frac{C}{(1+|z|)^{k+1}}e^{\sigma |\Im(z)|},\quad z\in\field{C},
\end{equation}
where $k\geq1$ and $g$ satisfying a similar estimate with the index $k'\geq1$, the truncation
error $T_N(t)$ of the partial sums in~\eqref{eq:OP-P-finite} satisfies the
estimate
\begin{equation}
|T_N(t)|\leq
\frac{2(\sigma/\pi)^{k}\mathcal{E}_k(t)}{(N+1)^{k}\sqrt{1-4^{-k}}}\frac{C'}{\sqrt{N+1}},
\end{equation}
for some constant $C'>0$, where
\begin{equation}
\mathcal{E}_{k}(t)=\sqrt{\frac{\pi}{\sigma}}
\left(\int_{\field{R}}s^{2k}|p(s+t)|^2ds\right)^{\frac{1}{2}}.
\end{equation}
\end{prop}
\begin{proof}
Using Cauchy-Schwartz
inequality in~\eqref{eq:trunc-OP-P}, we have
\begin{equation}\label{eq:TN}
\begin{split}
|T_N(t)|^2
\leq{\frac{\pi}{\sigma}}
\left(\sum_{|n|>N}\left|p\left(t+\frac{n\pi}{\sigma}\right)\right|^2\right)
\left(\sum_{|n|>N}|\hat{g}_n|^2\right).
\end{split}
\end{equation}
Observing that $t^{k}p(t)\in\fs{B}^2_{\sigma}$ and following 
Jagerman~\cite{J1966}, we can obtain the estimate
\begin{equation}
|T_N(t)|\leq
\frac{2(\sigma/\pi)^{k}\mathcal{E}_k(t)}{(N+1)^{k}\sqrt{1-4^{-k}}}
\left(\sum_{|n|>N}|\hat{g}_n|^2\right)^{1/2},
\end{equation}
for any $g\in\fs{B}_{\sigma}^2$. Noting that $g\in\fs{B}_{\sigma}^1$ and
\begin{equation*}
\begin{split}
\hat{g}_N
&=\int_0^{\infty}\psi_N(s)g(s)ds=g(N\pi/\sigma)-\int^0_{-\infty}\psi_N(s)g(s)ds,\\
|\hat{g}_N|&\leq|g(N\pi/\sigma)|
+\sqrt{\sigma/\pi^3}N^{-1}\int^0_{-\infty}|g(s)|ds=\bigO{N^{-1}},
\end{split}
\end{equation*}
therefore, $|\hat{g}_{\pm N}|=\bigO{N^{-1}}$. Now,
\[
\sum_{|n|>N}\frac{1}{N^2}\leq2\int_{N+1}^{\infty}\frac{ds}{s^2}=\frac{2}{N+1},
\]
so that
\[
\left(\sum_{|n|>N}|\hat{g}_n|^2\right)^{1/2}\leq\frac{C'}{\sqrt{N+1}}.
\]
for some constants $C'>0$. Plugging-in this estimate in~\eqref{eq:TN}, the result follows.
% \begin{equation*}
% |T_N(t)|\leq
% \frac{2(\sigma/\pi)^{k}\mathcal{E}_k(t)}{(N+1)^{k}\sqrt{1-4^{-k}}}\frac{C'}{\sqrt{N+1}},
% \end{equation*}
% for some constant $C'>0$.
\end{proof}
It is interesting to note that the truncation error does not improve by
strengthening the regularity condition of $g$ because $k'$ does not feature in the
estimate.

The representation in~\eqref{eq:OP-P-series} can also be used to define a linear
operator on $\ell^2$
\begin{equation}\label{eq:OP-P-discrete}
\begin{split}
\OP{P}[\vs{\alpha}](t)&=\sum_{n\in\field{Z}}
\sqrt{\frac{\pi}{\sigma}}p\left(t+\frac{n\pi}{\sigma}\right)\alpha_n,\quad\vs{\alpha}\in\ell^2.
\end{split}
\end{equation}
The infinite series converges absolutely and uniformly for $t\in\field{R}$ if
$p\in\fs{B}^2_{\sigma}$. Let us show that it defines a bounded linear operator
from $\ell^2\rightarrow\fs{B}_{\sigma}^2$ if $\rho\in\fs{L}^{\infty}$. Rewriting 
the infinite series in the Fourier domain, we have 
\begin{equation*}
\begin{split}
&\sum_{n\in\field{Z}}
\sqrt{\frac{\pi}{\sigma}}p\left(t+\frac{n\pi}{\sigma}\right)\alpha_n\\
&\qquad=\frac{1}{2\pi}\sqrt{\frac{\pi}{\sigma}}\int_{-\sigma}^{\sigma}
\left(\sum_{n\in\field{Z}}\alpha_ne^{i(n\pi/\sigma)\xi}\right)
\rho(\xi)e^{i\xi t}d\xi.
\end{split}
\end{equation*}
Using Plancheral's theorem and the fact that $\supp\rho\subset[-\sigma,\sigma]$, we have
\begin{equation*}
\begin{split}
&\left\|\sum_{n\in\field{Z}}
\sqrt{\frac{\pi}{\sigma}}p\left(t+\frac{n\pi}{\sigma}\right)\alpha_n\right\|_{\fs{L}^2}\\
&\qquad=\frac{1}{\sqrt{2\sigma}}\left\|
\left(\sum_{n\in\field{Z}}\alpha_ne^{i(n\pi/\sigma)\xi}\right)
\rho(\xi)\right\|_{\fs{L}^2(-\sigma,\sigma)}\\
&\qquad\leq\sqrt{{\pi}/{\sigma}}\|\rho\|_{\fs{L}^{\infty}(-\sigma,\sigma)}\|\vs{\alpha}\|_{\ell^2}.
\end{split}
\end{equation*}
These observations are summarized in the following proposition.
\begin{prop}
The operator $\OP{P}$ defined by~\eqref{eq:OP-P-discrete} with 
$p\in\fs{B}^{2}_{\sigma}$ defines a bounded linear operator from
$\ell^2\rightarrow\fs{B}_{\sigma}^2$ if
$\rho\in\fs{L}^{\infty}$.
\end{prop}

For the solution of the inverse scattering problem, the spectral properties of
the Hankel operators are relevant. Let us state the following result which
appears in somewhat general form in~\cite[Thm.~8.10]{P2003}.
\begin{theorem}
The Hankel operator $\OP{P}$ defined by~\eqref{eq:p-GLM-OP} with 
$p\in\fs{B}^{2}_{\sigma}$ is compact on $\fs{L}^2(\Omega_+)$, 
if there exists a function $\varphi\in\fs{C}(\field{R})$ with support in
$[-\sigma,\sigma]$ such that $p(t)$ agrees with $\fourier[\varphi](t)$ on
$\Omega_+$.
\end{theorem}

\begin{corr}
The Hankel operator $\OP{P}$ defined by~\eqref{eq:p-GLM-OP} 
is compact on $\fs{L}^2(\Omega_+)$ if $p\in\fs{B}^{1}_{\sigma}$ or,
alternatively, if $p(z)$ is an entire function such that
\begin{equation*}
|p(z)|\leq \frac{C}{(1+|z|)^{k+1}}e^{\sigma |\Im(z)|},\quad z\in\field{C},
\end{equation*}
holds for some $k>0$.
\end{corr}
\begin{proof}
In first cases, $p\in\fs{B}^{1}_{\sigma}$ ensures that
$\rho\in\fs{C}(\field{R})$. In the second case, the estimate simply ensures that
$p\in\fs{B}^{1}_{\sigma}$.
\end{proof}
%%%%%%%%%%%%%%%%%%%%%%%%%%%%%%%%%%%%%%%%%%%%%%%%%%%%%%%%%%%%%%%%%%%%%%%%%%%%%%%%%%%%%%
%%%%%%%%%%%%%%%%%%%%%%%%%%%%%%%%%%%%%%%%%%%%%%%%%%%%%%%%%%%%%%%%%%%%%%%%%%%%%%%%%%%%%%
%%%%%%%%%%%%%%%%%%%%%%%%%%%%%%%%%%%%%%%%%%%%%%%%%%%%%%%%%%%%%%%%%%%%%%%%%%%%%%%%%%%%%%
%%%%%%%%%%%%%%%%%%%%%%%%%%%%%%%%%%%%%%%%%%%%%%%%%%%%%%%%%%%%%%%%%%%%%%%%%%%%%%%%%%%%%%
\subsection{Sampling Approach to Inverse Scattering}\label{sec:sampling-approach}
The Hermitian conjugate of $\OP{P}$, denoted by $\OP{P}^{\dagger}$, with respect
to the inner product in $\fs{L}^2(\Omega_{+})$ works out to be
\begin{equation}
\OP{P}^{\dagger}[g](y)=\int_{0}^{\infty}p^*(y+s)g(s)ds.
\end{equation}
Define $\OP{K}=\OP{P}^{\dagger}\circ\OP{P}$, so that
\begin{equation}\label{eq:OP-K-GLM}
\begin{split}
\OP{K}[g](y)
&=\int_{0}^{\infty}ds\int_{0}^{\infty}dx\,p^*(y+s)p(s+x)g(x)\\
&=\int_{0}^{\infty}\mathcal{K}(y,x)g(x)dx,
\end{split}
\end{equation}
where the kernel function $\mathcal{K}(y,x;t)$ is given by
\begin{equation}
\mathcal{K}(y,x;\tau)
=\int_{0}^{\infty}ds\,p^*(y+s)p(s+x).
\end{equation}
The properties of the operator $\OP{K}$ can be deduced easily from that of
$\OP{P}$. If $p\in\fs{B}_{\sigma}^{1}$, the operator 
$\OP{K}$ defines a bounded linear operator on 
$\fs{L}^{\nu}(\Omega_{+}), (\nu=1,2)$ with the estimate
\begin{equation}
\|\OP{K}[g]\|_{\fs{L}^{\nu}(\Omega_{+})}
\leq\mathcal{I}^2_{1}\|g\|_{\fs{L}^{\nu}(\Omega_{+})}.
\end{equation}
Furthermore, it is a compact, self-adjoint and positive operator with respect to 
$\fs{L}^2(\Omega_{+})$. 

If $p\in\fs{B}_{\sigma}^{2}$ with $\rho\in\fs{L}^{\infty}$, $\OP{K}$ defines a
bounded, self-adjoint and positive linear operator on $\fs{L}^2(\Omega_{+})$. 
\begin{equation}
\|\OP{K}[g]\|_{\fs{L}^{\nu}(\Omega_{+})}
\leq\|\rho\|^2_{\fs{L}^{\infty}}\|g\|_{\fs{L}^{\nu}(\Omega_{+})}.
\end{equation}
The GLM equations in~\eqref{eq;GLM-start} can now be stated as
\begin{equation}\label{eq:fredholm}
{A}_j(y)={\Phi}_j(y)+\kappa\OP{K}[{A}_j](y),\quad j=1,2,
\end{equation}
which is a Fredholm integral equation of the second kind where
\begin{equation}
\left\{\begin{aligned}
&{\Phi}_1(y)=p^*(y),\\
&{\Phi}_2(y)=\kappa\OP{P}^{\dagger}[p](y).
\end{aligned}\right.
\end{equation}

\begin{theorem}\label{thm:OP-inv-K}
Let the operator $\OP{K}$ be defined by~\eqref{eq:OP-K-GLM} and $\OP{I}$ denote
the identity operator.
\begin{enumerate}[label=(\alph*)]
\item Let $p\in\fs{B}_{\sigma}^1$. If $\mathcal{I}_1<1$, then 
$(\OP{I}-\OP{K})^{-1}$ is a bounded linear operator 
on $\fs{L}^{\nu}(\Omega_{+})\,(\nu=1,2)$ with the estimate
\[
\|(\OP{I}-\OP{K})^{-1}\|_{\fs{L}^{\nu}(\Omega_{+})}
\leq\left(1-\mathcal{I}^2_1\right)^{-1}.
\]
\item Let $p\in\fs{B}_{\sigma}^2$ with $\rho\in\fs{L}^{\infty}$, then 
$(\OP{I}+\OP{K})^{-1}$ is a bounded 
linear operator on $\fs{L}^{2}(\Omega_{+})$ with 
\[
\|(\OP{I}+\OP{K})^{-1}\|_{\fs{L}^{2}(\Omega_{+})}\leq1.
\]
\end{enumerate}
\end{theorem}
\begin{proof}
To prove (a), we recall from the standard theory of linear operators, that 
if $\|\OP{K}\|_{\fs{L}^1(\Omega_+)}<1$, the operator 
$(\OP{I}-\OP{K})$ is invertible. The estimate for the inverse follows from
the observation that 
$\|\OP{K}\|_{\fs{L}^1(\Omega_+)}\leq\mathcal{I}^2_1$ 
when $p\in\fs{B}_{\sigma}^1$. 

To prove (b), let $p\in\fs{B}^2_{\sigma}$ with $\rho\in\fs{L}^{\infty}$. Under 
this condition, $\OP{K}$ exists as a bounded linear operator on 
$\fs{L}^2(\Omega_{+})$. It is easy to verify that $(\OP{I}+\OP{K})$ is 
positive and, as a result, bounded from below:
\[
\|(\OP{I}+\OP{K})[f]\|_{\fs{L}^2(\Omega_{+})}\geq \|f\|_{\fs{L}^2(\Omega_{+})},
\]
for every $f\in\fs{L}^2(\Omega_{+})$. Consequently,
$\Ker(\OP{I}+\OP{K})=\{0\}$. This establishes that $(\OP{I}+\OP{K})$ has a bounded inverse 
on its range which is closed~\cite[Chap.~1, Thm.~1.2]{K2012}. Now, noting
$(\OP{I}+\OP{K})$ is self-adjoint, we have
\[
{\Ran(\OP{I}+\OP{K})}
=[\Ker(\OP{I}+\OP{K})]^{\perp}=\{0\}^{\perp}=\fs{L}^2(\Omega_+).
\] 
Let $g=(\OP{I}+\OP{K})[f]$, then  
\begin{multline*}
\|(\OP{I}+\OP{K})^{-1}[g]\|_{\fs{L}^{2}(\Omega_{+})}\leq\|f\|_{\fs{L}^2(\Omega_{+})}\\
\leq\|(\OP{I}+\OP{K})[f]\|_{\fs{L}^2(\Omega_{+})}=\|g\|_{\fs{L}^2(\Omega_{+})}
\end{multline*}
yields $\|(\OP{I}+\OP{K})^{-1}\|_{\fs{L}^{2}(\Omega_{+})}\leq1$.
\end{proof}
\begin{rem}
When $\kappa=1$, it is known from the analysis of the Zakharov-Shabat scattering
problem that $|\rho(\xi)|<1$ for $\xi\in\field{R}$. Therefore, it follows that for 
$p\in\fs{B}_{\sigma}^2$, the operator $(\OP{I}-\OP{K})$ is
always invertible in $\fs{L}^{2}(\Omega_{+})\,(\nu=1,2)$ provided the reflection
coefficient is admissible.
\end{rem}
Turning to the discrete representation of GLM equations, let us introduce
\begin{equation}
\alpha^{(j)}_n=\int_{0}^{\infty}A_j(s)\psi_n(s)ds,\quad j=1,2,
\end{equation}
so that the GLM equations can be written as
\begin{equation}\label{eq:GLM-discrete}
\begin{split}
&\kappa A^{*}_2(t) = \sum_{n\in\field{Z}}\sqrt{\frac{\pi}{\sigma}}
p\left(t+\frac{n\pi}{\sigma}\right)\alpha^{(1)}_n,\\
&A_1^{*}(t) = p(t)+\sum_{n\in\field{Z}}\sqrt{\frac{\pi}{\sigma}}
p\left(t+\frac{n\pi}{\sigma}\right)\alpha^{(2)}_n.
\end{split}
\end{equation}
Define
\begin{equation}\label{eq:p-M-mat}
\begin{split}
&\hat{p}_l=\int_{0}^{\infty}p(s)\psi_l(s)ds,\\
&\mathcal{M}_{lm}
=\sqrt{\frac{\pi}{\sigma}}\int_{0}^{\infty}p\left(s+\frac{m\pi}{\sigma}\right)\psi_l(s)dt,
\end{split}
\end{equation}
so that
\begin{equation}\label{eq:GLM-discrete-full}
\left\{\begin{aligned}
&\kappa \alpha^{(2)*}_l  = \sum_{m\in\field{Z}}\mathcal{M}_{lm}\alpha^{(1)}_m,\\
&\alpha^{(1)*}_l = \hat{p}_l+\sum_{m\in\field{Z}}\mathcal{M}_{lm}\alpha^{(2)}_m,
\end{aligned}\right.\quad l\in\field{Z}.
\end{equation}
which can be stated in a compact form by introducing the infinite column vectors
$\vs{\alpha}_j=(\alpha^{(j)}_n)_{n\in\field{Z}}$ and
$\hat{\vv{p}}=(\hat{p}_n)_{n\in\field{Z}}$:
\begin{equation}
\kappa\vs{\alpha}^{*}_2=\mathcal{M}\vs{\alpha}_1,\quad
\vs{\alpha}^{*}_1=\hat{\vv{p}}+\mathcal{M}\vs{\alpha}_2.
\end{equation}
Before attempting to solve the GLM equations, we analyze the properties of the
``mass'' matrix $\mathcal{M}$. Using sampling expansions, we can write
\begin{equation}
\mathcal{M}_{lm}
={\frac{\pi}{\sigma}}\sum_{n\in\field{Z}}\mathcal{Q}_{ln}p\left((m+n)\frac{\pi}{\sigma}\right)
=\sum_{n\in\field{Z}}\mathcal{Q}_{ln}\mathcal{P}_{nm}
\end{equation}
where $\mathcal{P}$ is the Hankel matrix defined by 
\begin{equation}\label{eq:P-mat}
\mathcal{P}_{nm}={\frac{\pi}{\sigma}}p\left((m+n)\frac{\pi}{\sigma}\right),
\end{equation}
and, $\mathcal{Q}$ is a real symmetric matrix defined by
\begin{equation}\label{eq:Q-mat}
\mathcal{Q}_{nl}=\int_{0}^{\infty}\psi_n(s)\psi_l(s)ds,
\end{equation}
which we would refer to as the \emph{quadrature matrix}. An estimate for the 
values of each its entries can be easily obtained using the 
Cauchy-Schwartz inequality:
\begin{equation}
|\mathcal{Q}_{nl}|\leq
\|\psi_n\|_{\fs{L}^2(\Omega_{+})}
\|\psi_l\|_{\fs{L}^2(\Omega_{+})}<1.
\end{equation}
It turns out that the entries of $\mathcal{Q}$ can be
computed exactly in terms of the Sine and Cosine integrals (see
Appendix~\ref{sec:q-mat}):
\begin{equation}
\mathcal{Q}_{ln}=
\begin{cases}
\frac{1}{2}-\frac{1}{\pi}\Si(2|n|\pi), & l=n\\
\frac{(-1)^{l+n}}{2\pi^2(l-n)}\left[\Cin(2l\pi)-\Cin(2n\pi)\right],& l\neq n.
\end{cases}
\end{equation}
Using this quadrature matrix, we can also write
\begin{equation}
\hat{p}_{l}
=\sqrt{\frac{\pi}{\sigma}}\sum_{n\in\field{Z}}\mathcal{Q}_{ln}p\left(\frac{n\pi}{\sigma}\right)
=\sum_{n\in\field{Z}}\mathcal{Q}_{ln}{p}_{n},
\end{equation}
where $p_n=\sqrt{\pi/\sigma}p(n\pi/\sigma)$. 
%%%%%%%%%%%%%%%%%%%%%%%%%%%%%%%%%%%%%%%%%%%%%%%%%%%%%%%%%%%%%%%%%%%%%%%%%%%%%%%%%%%%%%%%%
\begin{prop}
Assume $p\in\fs{B}^2_{\sigma}$ with $\rho\in\fs{L}^{\infty}$ and let the infinite matrices $\mathcal{M}$, 
$\mathcal{P}$ and $\mathcal{Q}$ be defined by
~\eqref{eq:p-M-mat},~\eqref{eq:P-mat} and~\eqref{eq:Q-mat}, respectively.
\begin{enumerate}[label=(\alph*)]
\item The matrix $\mathcal{M}$ defines a bounded linear operator on $\ell^2$.
\item The real symmetric matrix $\mathcal{Q}$ is positive definite and defines a
bounded positive linear operator on $\ell^2$.
\item The complex symmetric matrix $\mathcal{P}$ defines a bounded Hankel matrix on $\ell^2$.
\end{enumerate}
\end{prop}
\begin{proof}
Let $p\in\fs{B}^2_{\sigma}$. Then, for $l,m\in\field{Z}$, we have
\begin{equation*}
\begin{split}
\mathcal{M}_{lm}&=\sqrt{\frac{\pi}{\sigma}}\OP{S}\left[p\left(\cdot+\frac{m\pi}{\sigma}\right)\right]
\left(\frac{l\pi}{\sigma}\right)\\
&=\sqrt{\frac{\pi}{\sigma}}\int_{0}^{\infty}
\psi_0\left(s-\frac{l\pi}{\sigma}\right)
p\left(s+\frac{m\pi}{\sigma}\right)ds.
\end{split}
\end{equation*}
For any $\vv{\alpha}\in\ell^2$, let us note that 
\[
\mathcal{C}(s)\equiv\sqrt{\frac{\pi}{\sigma}}
\sum_{m\in\field{Z}}\alpha_mp\left(s+\frac{m\pi}{\sigma}\right)\in\fs{B}^2_{\sigma},
\]
so that
\begin{equation*}
\begin{split}
\left(\sum_{l\in\field{Z}}\left|\sum_{m\in\field{Z}}\mathcal{M}_{lm}\alpha_m\right|^2\right)^{\frac{1}{2}}
&=\left(\sum_{l\in\field{Z}}
\left|\int^{\infty}_{0}\psi_0\left(s-\frac{l\pi}{\sigma}\right)
\mathcal{C}(s)ds\right|^2\right)^{\frac{1}{2}}\\
&\leq(1+\pi)\|\OP{S}\|_{\fs{L}^2}\|\mathcal{C}\|_{\fs{L}^2}\\
&\leq(1+\pi)({\pi}/{\sigma})\|\rho\|_{\fs{L}^{\infty}}\|\vs{\alpha}\|_{\ell^2},
\end{split}
\end{equation*}
which follows from~\eqref{eq:B-spaces-norm}, 
Prop.~\ref{lemma:bounded-hankel} and~\eqref{eq:OP-S-norm}. Using similar arguments 
for $\mathcal{Q}_{lm}$, we have
\begin{equation*}
\begin{split}
\left(\sum_{l\in\field{Z}}\left|\sum_{m\in\field{Z}}\mathcal{Q}_{lm}\alpha_m\right|^2\right)^{\frac{1}{2}}
&\leq(1+\pi)({\pi}/{\sigma})\|\vs{\alpha}\|_{\ell^2}.
\end{split}
\end{equation*}
This shows that the symmetric real matrix $\mathcal{Q}$ is bounded.
It also turns out to be a positive definite matrix:
\begin{equation*}
\vs{\alpha}^{\dagger}\mathcal{Q}\vs{\alpha}
=\int_{0}^{\infty}\left|\sum_{n\in\field{Z}}\alpha_n\psi_n(s)\right|^2ds
\leq\|\vs{\alpha}\|^2_{\ell^2},
\end{equation*}
yielding $\|\mathcal{Q}\|_{\ell^2}\leq1$.

In order to prove the last statement, let us note that the symbol of the 
Hankel matrix $\mathcal{P}$ can be worked out to be
$\rho(\sigma\theta/\pi),\,\theta\in[-\pi,\pi]$; therefore, on account of
$\rho\in\fs{L}^{\infty}$, the matrix $\mathcal{P}$ turns out to be a bounded 
Hankel matrix.
%  therefore,
% \begin{equation}
% \|\mathcal{M}\|_{\ell^2}\leq\|\mathcal{P}\|_{\ell^2}.
% \end{equation}
\end{proof}
The solution of the coupled system in~\eqref{eq:GLM-discrete-full} can be obtained by defining the
$2\times2$ block matrix equation
\begin{equation}
\begin{pmatrix}
I & -\mathcal{M}\\
-\kappa\mathcal{M}^* & I
\end{pmatrix}
\begin{pmatrix}
\vs{\alpha}^*_1\\
\vs{\alpha}_2
\end{pmatrix}
=
\begin{pmatrix}
\hat{\vv{p}}\\
\vv{0}
\end{pmatrix}.
\end{equation}
Using $\mathcal{M}=\mathcal{Q}\mathcal{P}$, we may symmetrize the linear
system by introducing $\vs{\alpha}_j=\mathcal{Q}^{1/2}\vs{\beta}_j,\,j=1,2$ so
that
\begin{equation*}
\begin{pmatrix}
I & -\mathcal{Q}^{1/2}\mathcal{P}\mathcal{Q}^{1/2}\\
-\kappa\mathcal{Q}^{1/2}\mathcal{P}^*\mathcal{Q}^{1/2} & I
\end{pmatrix}
\begin{pmatrix}
\vs{\beta}^*_1\\
\vs{\beta}_2
\end{pmatrix}
=
\begin{pmatrix}
\mathcal{Q}^{1/2}{\vv{p}}\\
\vv{0}
\end{pmatrix},
\end{equation*}
where the column vector
$\vv{p}=\sqrt{\pi/\sigma}(p(n\pi/\sigma))_{n\in\field{Z}}$. Putting
$\mathcal{G}=\mathcal{Q}^{1/2}\mathcal{P}\mathcal{Q}^{1/2}$, we have
\begin{equation}
\begin{pmatrix}
                 I & -\mathcal{G}\\
-\kappa\mathcal{G}^{\dagger} &  I
\end{pmatrix}
\begin{pmatrix}
\vs{\beta}^*_1\\
\vs{\beta}_2
\end{pmatrix}
=
\begin{pmatrix}
\mathcal{Q}^{1/2}{\vv{p}}\\
\vv{0}
\end{pmatrix}.
\end{equation}
The block matrix on the right hand side is Hermitian for $\kappa=+1$ and it can
be reduced to a Hermitian form for the case $\kappa=-1$ by rearranging:
\begin{equation}
\begin{pmatrix}
                 -I & \mathcal{G}\\
\mathcal{G}^{\dagger} &            I
\end{pmatrix}
\begin{pmatrix}
\vs{\beta}^*_1\\
\vs{\beta}_2
\end{pmatrix}
=
\begin{pmatrix}
-\mathcal{Q}^{1/2}{\vv{p}}\\
\vv{0}
\end{pmatrix}.
\end{equation}
On eliminating $\vs{\beta}_2$, we have
\begin{equation}
(I-\kappa\mathcal{G}^{\dagger}\mathcal{G})\vs{\beta}_1
=\mathcal{Q}^{1/2}{\vv{p}}^*.
\end{equation}
Setting $\kappa=+1$, in order for $(I-\mathcal{G}^{\dagger}\mathcal{G})$ to be
invertible, it suffices to have $\|\mathcal{P}\|_{\ell^2}<1$.  
\begin{rem}
The symbol of the Hankel matrix $\mathcal{P}$ can be worked out to be
$\rho(\sigma\theta/\pi),\,\theta\in[-\pi,\pi]$; therefore, the requirement 
$\|\mathcal{P}\|_{\ell^2}<1$ can be fulfilled if $|\rho|<1$ for all 
$\theta\in[-\pi,\pi]$.
\end{rem}
Setting $\kappa=-1$, the infinite matrix $\mathcal{G}^{\dagger}\mathcal{G}$
defines a positive linear operator; therefore, $(I-\mathcal{G}^{\dagger}\mathcal{G})$ 
invertible. Finally, we have the potential given by
\begin{equation}
\begin{split}
q(0) &= -2A_1(0)=-2p^*(0)
-2\sum_{n\in\field{Z}}\sqrt{\frac{\pi}{\sigma}}
p^*\left(\frac{n\pi}{\sigma}\right)\alpha^{(2)*}_n\\
&=-2p^*(0)
-2\vv{p}^{\dagger}\vs{\alpha}^{*}_2\\
&=-2p^*(0)
-2\kappa\vv{p}^{\dagger}\mathcal{Q}^{1/2}
\mathcal{G}(I-\kappa\mathcal{G}^{\dagger}\mathcal{G})^{-1}\mathcal{Q}^{1/2}{\vv{p}}^*,
\end{split}
\end{equation}
together with its $\fs{L}^2$-norm on $\Omega_{+}$ as
\begin{equation}
\begin{split}
\|q\|^2_{\fs{L}^2(\Omega_{+})}
&=2\kappa A_2(0)=2\sum_{n\in\field{Z}}\sqrt{\frac{\pi}{\sigma}}
p^*\left(\frac{n\pi}{\sigma}\right)\alpha^{(1)*}_n\\
&=2\vv{p}^{\dagger}\vs{\alpha}^{*}_1
=2\vv{p}^{\dagger}\mathcal{Q}^{1/2}
(I-\kappa\mathcal{G}\mathcal{G}^{\dagger})^{-1}\mathcal{Q}^{1/2}{\vv{p}}.
\end{split}
\end{equation}

%%%%%%%%%%%%%%%%%%%%%%%%%%%%%%%%%%%%%%%%%%%%%%%%%%%%%%%%%%%%%%%%%%%%%%%%%%%%%%%
%%%%%%%%%%%%%%%%%%%%%%%%%%%%%%%%%%%%%%%%%%%%%%%%%%%%%%%%%%%%%%%%%%%%%%%%%%%%%%%
%%%%%%%%%%%%%%%%%%%%%%%%%%%%%%%%%%%%%%%%%%%%%%%%%%%%%%%%%%%%%%%%%%%%%%%%%%%%%%%
Turning to the numerical aspects, let us introduce the truncated version of the
GLM equations that can be implemented as a numerical scheme. To this end, define
\begin{equation}\label{eq:GLM-discrete1}
\begin{split}
&\kappa A^{(N)*}_2(t) = \sum_{|n|\leq N}\sqrt{\frac{\pi}{\sigma}}
p\left(t+\frac{n\pi}{\sigma}\right)\alpha^{(1,N)}_n,\\
&A_1^{(N)*}(t) = p(t)+\sum_{|n|\leq N}\sqrt{\frac{\pi}{\sigma}}
p\left(t+\frac{n\pi}{\sigma}\right)\alpha^{(2,N)}_n.
\end{split}
\end{equation}
so that 
\begin{equation}
\left\{\begin{aligned}
&\kappa \alpha^{(2,N)*}_l  = \sum_{|m|\leq N}\mathcal{M}_{lm}\alpha^{(1,N)}_m,\\
&\alpha^{(1,N)*}_l = p_l+\sum_{|m|\leq N}\mathcal{M}_{lm}\alpha^{(2,N)}_m,
\end{aligned}\right.\quad |l|\leq N.
\end{equation}
or, equivalently,
\begin{equation}
\kappa\vs{\alpha}^{(N)*}_2=\mathcal{M}_N\vs{\alpha}^{(N)}_1,\quad
\vs{\alpha}^{(N)*}_1=\hat{\vv{p}}_N+\mathcal{M}_N\vs{\alpha}^{(N)}_2,
\end{equation}
which simplifies to
\begin{equation}
(I_N-\kappa \mathcal{M}^*_N\mathcal{M}_N)\vs{\alpha}^{(N)}_1=\hat{\vv{p}}^*_N.
\end{equation}
The potential is then obtained from
\begin{equation}
\begin{split}
q(0) &\approx
-2p^*(0)
-2\vv{p}^{\dagger}_N\vs{\alpha}^{(N)*}_2,
\end{split}
\end{equation}
together with its $\fs{L}^2$-norm on $\Omega_{+}$ as
\begin{equation}
\begin{split}
\|q\|^2_{\fs{L}^2(\Omega_{+})}
&\approx
2\vv{p}^{\dagger}_N\vs{\alpha}^{(N)*}_1.
\end{split}
\end{equation}
Let us note that by introducing a truncated quadrature matrix 
$\mathcal{Q}_N$, the linear system can be symmetrized by putting
$\mathcal{G}_N=\mathcal{Q}_N^{1/2}\mathcal{P}_N\mathcal{Q}_N^{1/2}$ 
in the same manner as described earlier. For the moment, we hold off 
the symmetrization procedure as it involves quadrature errors on account of 
the fact that the quadrature formula with finite $\mathcal{Q}_N$ is 
not exact. Now we turn to the convergence analysis of the numerical procedure 
described above. To this end, we consider the total numerical 
error $R^{(N)}_j(t)\,(j=1,2)$ given by
\begin{equation}\label{eq:GLM-discrete-error}
\begin{split}
\kappa R^{(N)*}_2(t) 
&= \sum_{|n|\leq N}\sqrt{\frac{\pi}{\sigma}}
p\left(t+\frac{n\pi}{\sigma}\right)\left(\alpha^{(1)}_n-\alpha^{(1,N)}_n\right)\\
&\qquad+\sum_{|n|>N}\sqrt{\frac{\pi}{\sigma}}
p\left(t+\frac{n\pi}{\sigma}\right)\alpha^{(1)}_n,\\
R_1^{(N)*}(t) 
&= \sum_{|n|\leq N}\sqrt{\frac{\pi}{\sigma}}
p\left(t+\frac{n\pi}{\sigma}\right)\left(\alpha^{(2)}_n-\alpha^{(2,N)}_n\right)\\
&\qquad+\sum_{|n|>N}\sqrt{\frac{\pi}{\sigma}}
p\left(t+\frac{n\pi}{\sigma}\right)\alpha^{(2)}_n.
\end{split}
\end{equation}
Define
\begin{equation}\label{eq:GLM-discrete-EN}
E_N(t)=\sqrt{\frac{\pi}{\sigma}}
\left(\sum_{|n|>N}\left|p\left(t+\frac{n\pi}{\sigma}\right)\right|^2\right)^{\frac{1}{2}},
\end{equation}
then, it is easy work out the estimates
\begin{equation}\label{eq:error-estimate1}
\begin{split}
|R^{(N)}_2(t)|
&\leq(1+\pi)\sqrt{{\pi}/{\sigma}}\|p\|_{\fs{L}^2}\|\vs{\alpha}_1-\vs{\alpha}^{(N)}_1\|_{\ell^2}\\
&\qquad+ E_N(t)\left(\sum_{|n|>N}|\alpha^{(1)}_n|^2\right)^{\frac{1}{2}},\\
|R_1^{(N)}(t)|
&\leq(1+\pi)\sqrt{{\pi}/{\sigma}}\|p\|_{\fs{L}^2}\|\vs{\alpha}_2-\vs{\alpha}^{(N)}_2\|_{\ell^2}\\
&\qquad+ E_N(t)\left(\sum_{|n|>N}|\alpha^{(2)}_n|^2\right)^{\frac{1}{2}}.
\end{split}
\end{equation}
By a slight abuse of notations, let the vectors $\vs{\alpha}^{(N)*}_j\,(j=1,2)$
and $\hat{\vs{p}}_N$ represent infinite dimensional vectors with entries corresponding to $|n|>N$
taken to be identically zero. Then, it is straightforward to work out
\begin{equation}\label{eq:discrete-error}
\begin{split}
\begin{pmatrix}
\vs{\alpha}^*_1 - \vs{\alpha}^{(N)*}_1\\
\vs{\alpha}_2   - \vs{\alpha}^{(N)}_2
\end{pmatrix}
&=
\begin{pmatrix}
I & -\mathcal{M}\\
-\kappa\mathcal{M}^* & I
\end{pmatrix}^{-1}
\begin{pmatrix}
\hat{\vv{p}}-\hat{\vv{p}}_N\\
\vv{0}
\end{pmatrix}\\
&=
\begin{pmatrix}
(I-\kappa\mathcal{M}\mathcal{M}^*)^{-1}[\hat{\vv{p}}-\hat{\vv{p}}_N]\\
\kappa \mathcal{M}^*(I-\kappa\mathcal{M}\mathcal{M}^*)^{-1}[\hat{\vv{p}}-\hat{\vv{p}}_N]
\end{pmatrix}.
\end{split}
\end{equation}
Therefore, under the conditions that ensure $(I-\kappa\mathcal{M}\mathcal{M}^*)$
is invertible, it follows that, for some constant $C>0$, the estimates
\begin{equation}\label{eq:error-estimate2}
\|\vs{\alpha}_j-\vs{\alpha}^{(N)*}_j\|_{\ell^2}\leq
C\|\hat{\vv{p}}-\hat{\vv{p}}_N\|_{\ell^2},\quad j=1,2,
\end{equation}
hold. The estimates~\eqref{eq:error-estimate1} and~\eqref{eq:error-estimate2},
allow us to conclude that the truncated system converges to the true solution
under the aforementioned conditions. It is possible to make precise statement about
the rate of convergence with respect to the number of
basis functions $2N+1$ of the truncated system if we strengthen the regularity
condition on $p$: 
\begin{prop}\label{eq:convg-WKS}
For $p$ satisfying an estimate of the form 
\begin{equation}
|p(z)|\leq\frac{C}{(1+|z|)^{2}}e^{\sigma |\Im(z)|},\quad z\in\field{C},
\end{equation}
we have for fixed $t\in\field{R}$, 
\[
|R^{(N)}_j(t)|=\bigO{N^{-1/2}},\quad j=1,2,
\]
where $2N+1$ is the number of basis functions used in the sampling expansion. 
\end{prop}
\begin{proof}
It is easy to see that the solution of the GLM equations exists under the 
conditions prescribed in the proposition. Under the same conditions, following
the methods discussed in the last section, it is also easy to show that  
\begin{equation*}
\|\hat{\vv{p}}-\hat{\vv{p}}_N\|_{\ell^2}=\bigO{N^{-1/2}},
\end{equation*}
and $E_N(t)=\bigO{N^{-1}}$ (see Prop.~\ref{prop:jagerman-type}) so that 
from~\eqref{eq:error-estimate1} and~\eqref{eq:error-estimate2} the result follows.
\end{proof}
%%%%%%%%%%%%%%%%%%%%%%%%%%%%%%%%%%%%%%%%%%%%%%%%%%%%%%%%%%%%%%%%%%%%%%%%%%%%%%%%%%%%%%%%%%%%%%%%%%%%
%%%%%%%%%%%%%%%%%%%%%%%%%%%%%%%%%%%%%%%%%%%%%%%%%%%%%%%%%%%%%%%%%%%%%%%%%%%%%%%%%%%%%%%%%%%%%%%%%%%%
%%%%%%%%%%%%%%%%%%%%%%%%%%%%%%%%%%%%%%%%%%%%%%%%%%%%%%%%%%%%%%%%%%%%%%%%%%%%%%%%%%%%%%%%%%%%%%%%%%%%
\subsubsection{Quadrature Errors}
Let us observe that the entries of the mass matrix $\mathcal{M}$ require a quadrature method 
which works well on infinite domains and is capable of providing higher orders
of convergence depending on the regularity of $p$. The quadrature matrix
$\mathcal{Q}$ exploits the sampling expansion to achieve this goal. Fortunately,
the need for a numerical quadrature is avoided by computing the
integrals exactly. However, truncation of this quadrature matrix introduces
numerical errors in computing $\mathcal{M}$. Note that the same difficulties
also arise in the computation of the vector $\hat{\vv{p}}$. In order to quantify
these errors, let us consider $2M+1$ number of basis functions and define
\begin{equation}
\begin{split}
p^{(M)}_l&=\sqrt{\frac{\pi}{\sigma}}
\sum_{|n|\leq M}\mathcal{Q}_{ln}p\left(\frac{n\pi}{\sigma}\right),\\
\mathcal{M}^{(M)}_{ln}
&={\frac{\pi}{\sigma}}\sum_{|m|\leq M}
\mathcal{Q}_{lm}p\left((m+n)\frac{\pi}{\sigma}\right).
\end{split}
\end{equation}
The total numerical error as a result of truncation of the linear system as
well as the quadrature matrix can be written as
\begin{equation}
\begin{pmatrix}
\vs{\alpha}^*_1 - \vs{\alpha}^{(N,M)*}_1\\
\vs{\alpha}_2   - \vs{\alpha}^{(N,M)}_2
\end{pmatrix}
=
\begin{pmatrix}
\vs{\alpha}^*_1 - \vs{\alpha}^{(N)*}_1\\
\vs{\alpha}_2   - \vs{\alpha}^{(N)}_2
\end{pmatrix}
+
\begin{pmatrix}
\vs{\alpha}^{(N)*}_1 - \vs{\alpha}^{(N,M)*}_1\\
\vs{\alpha}^{(N)}_2   - \vs{\alpha}^{(N,M)}_2
\end{pmatrix},
\end{equation}
where $\vs{\alpha}^{(N,M)*}_j$ refers to the solution of the linear system
constructed using $(2N+1)\times (2M+1)$ quadrature matrix. The relevant
quantities of this linear system are labeled as $\mathcal{M}^{(M)}_N$ and 
$\hat{\vv{p}}^{(M)}_N$ whose meanings are self-evident. We have already dealt
with the difference $\vs{\alpha}_j - \vs{\alpha}^{(N)}_j$; let us then turn
to the difference 
$\vs{\alpha}^{(N)}_j-\vs{\alpha}^{(N,M)}_j, j=1,2,$ which is given by
\begin{multline}
\begin{pmatrix}
\vs{\alpha}^{(N)*}_1 - \vs{\alpha}^{(N,M)*}_1\\
\vs{\alpha}^{(N)}_2   - \vs{\alpha}^{(N,M)}_2
\end{pmatrix}=
\begin{pmatrix}
I & -\mathcal{M}^{(M)}_N\\
-\kappa\mathcal{M}^{(M)*}_N & I
\end{pmatrix}^{-1}\\
\times
\begin{pmatrix}
(\mathcal{M}_N-\mathcal{M}_N^{(M)})\vs{\alpha}^{(N)}_2+(\hat{\vv{p}}_N-\hat{\vv{p}}^{(M)}_N)\\
\kappa(\mathcal{M}_N^*-\mathcal{M}_N^{(M)*})\vs{\alpha}^{(N)*}_1
\end{pmatrix}.
\end{multline}
In view of the expression above, it suffices to estimate
$(\mathcal{M}_N-\mathcal{M}_N^{(M)})$ and
$(\hat{\vv{p}}_N-\hat{\vv{p}}^{(M)}_N)$:
\begin{lemma}\label{lemma:Q-errors}
Let $p$ satisfy an estimate of the form 
\begin{equation}
|p(z)|\leq\frac{C}{(1+|z|)^{k+1}}e^{\sigma |\Im(z)|},\quad z\in\field{C},
\end{equation}
where $k\geq1$. Let the quadrature matrix $\mathcal{Q}$ be truncated to the size 
$(2N+1)\times (2M+1)$.
\begin{enumerate}[label=(\alph*)]
\item Let $\hat{\vv{p}}^{(M)}_N$ denote the approximation to $\hat{\vv{p}}_N$
using the quadrature matrix $\mathcal{Q}$, then the estimate
\begin{equation}
\|\hat{\vv{p}}_N-\hat{\vv{p}}^{(M)}_N\|_{\ell^2}\leq
\frac{2(\sigma/\pi)^{k}\mathcal{E}_k}{(M+1)^{k}\sqrt{1-4^{-k}}},
\end{equation}
holds where 
\begin{equation*}
\mathcal{E}_k=\sqrt{\frac{\pi}{\sigma}}
\left(\int_{\field{R}}s^{2k}|p(s)|^2ds\right)^{\frac{1}{2}}.
\end{equation*}
\item Let $\mathcal{M}^{(M)}_N$ denote the approximation to $\mathcal{M}_{N}$ using 
the quadrature matrix $\mathcal{Q}$, then the estimate
\begin{equation}
\|\mathcal{M}_{N}-\mathcal{M}^{(M)}_N\|_{\ell^2}
\leq\frac{2(\sigma/\pi)^{k-1}\|\partial_{\xi}^k\rho\|_{\fs{L}^{\infty}(-\sigma,\sigma)}}
{(M+1)^{k}\sqrt{1-4^{-k}}},
\end{equation}
holds.
\end{enumerate}
\end{lemma}
\begin{proof}
To prove (a), we observe that
\begin{equation}
\|\hat{\vv{p}}_N-\hat{\vv{p}}^{(M)}_N\|_{\ell^2}\leq\|\mathcal{Q}\|_{\ell^2}
\sqrt{\frac{\pi}{\sigma}}
\left(\sum_{|m|>M}\left|p\left(\frac{n\pi}{\sigma}\right)\right|^2\right)^{\frac{1}{2}}.
\end{equation}
The result then follows by noting that $\|\mathcal{Q}\|_{\ell^2}<1$ and using
Jagerman's estimate~\cite{J1966} for the remaining expression above.

To prove (b), let $\vs{\alpha}\in\ell^2$, then
\begin{equation*}
\begin{split}
&\|[\mathcal{M}_{N}-\mathcal{M}^{(M)}_N]\vs{\alpha}\|_{\ell^2}\\
&\leq\|\mathcal{Q}\|_{\ell^2}
{\frac{\pi}{\sigma}}\left[\sum_{|m|>M}
\left|\sum_{|n|\leq N}p\left((m+n)\frac{\pi}{\sigma}\right)\alpha_n\right|^2\right]^{\frac{1}{2}}.
\end{split}
\end{equation*}
Define
\begin{equation}
\mathcal{C}(s)=\sum_{n\in\field{Z}}p\left(s+\frac{n\pi}{\sigma}\right)\alpha_n,
\end{equation}
then $s^k\mathcal{C}(s)\in\fs{L}^2$ and 
\begin{equation}
\mathcal{E}_k
=\left(\int_{\field{R}}s^{2k}|\mathcal{C}(s)|^2ds\right)^{\frac{1}{2}}
\leq \|\partial_{\xi}^k\rho\|_{\fs{L}^{\infty}(-\sigma,\sigma)}
\|\vs{\alpha}\|_{\ell^2}.
\end{equation}
Now, observing that $\|\mathcal{Q}\|_{\ell^2}<1$ and
\begin{equation}
\begin{split}
\left[\sum_{|m|>M}
\left|\sum_{|n|\leq N}p\left((m+n)\frac{\pi}{\sigma}\right)\alpha_n\right|^2\right]^{\frac{1}{2}}
&\leq\left[\sum_{|m|>M}\left|\mathcal{C}\left(\frac{m\pi}{\sigma}\right)\right|^2\right]^{\frac{1}{2}},
\end{split}
\end{equation}
the result follows by using  Jagerman's estimate~\cite{J1966}.
\end{proof}
Let us conclude this section with the following remark. Based on the estimates
obtained above, it is clear that choice $M=N$ does not alter the rate of
convergence. Besides, this choice makes it possible to symmetrize the truncated linear
system which ensures that it is well conditioned. In particular, the uniform boundedness
of the inverse of
\[
\begin{pmatrix}
I & -\mathcal{M}^{(N)}_N\\
-\kappa\mathcal{M}^{(N)*}_N & I
\end{pmatrix}
\]
can be established exactly in the manner we treated the infinite case.
%%%%%%%%%%%%%%%%%%%%%%%%%%%%%%%%%%%%%%%%%%%%%%%%%%%%%%%%%%%%%%%%%%%%%%%%%%%%%%%%%%%%%%
%%%%%%%%%%%%%%%%%%%%%%%%%%%%%%%%%%%%%%%%%%%%%%%%%%%%%%%%%%%%%%%%%%%%%%%%%%%%%%%%%%%%%%
%%%%%%%%%%%%%%%%%%%%%%%%%%%%%%%%%%%%%%%%%%%%%%%%%%%%%%%%%%%%%%%%%%%%%%%%%%%%%%%%%%%%%%
%%%%%%%%%%%%%%%%%%%%%%%%%%%%%%%%%%%%%%%%%%%%%%%%%%%%%%%%%%%%%%%%%%%%%%%%%%%%%%%%%%%%%%
\subsection{Modified Sampling Approach}\label{sec:HT}
The slow convergence of the sampling series motivates us to consider modified
versions of the sampling theorem which facilitate faster convergence at the cost
of sampling beyond the Nyquist rate. A modified version of the sampling series
was proposed by Helms and Thomas~\cite{HT1962,J1966} which can be described as
follows: Introducing the bandlimiting parameter $\sigma'$ and $\delta\in(0,1)$
such that
\begin{equation}
\sigma' = \sigma/(1-\delta),
\end{equation}
where $\sigma$ defines the support for the reflection coefficient
$\rho$, i.e., $\supp\rho\subset[-\sigma,\sigma]$. Let us define
\begin{equation}
\left\{\begin{aligned}
&\sigma'_-=(1-\delta)\sigma'=\sigma,\\
&\sigma'_+=(1+\delta)\sigma'=\left(\frac{1+\delta}{1-\delta}\right)\sigma.
\end{aligned}\right.
\end{equation}
The sinc basis functions are then modified by a multiplier of the form
\begin{equation}
\theta(t)=\left[\frac{\sin\left(\frac{\delta}{m}\sigma'
t\right)}{\left(\frac{\delta}{m}\sigma' t\right)}\right]^m,\quad m\in\field{Z}_+,
\end{equation}
in order to accelerate the convergence of the sampling series. Let us define the 
basis functions
\begin{equation}\label{eq:new-basis}
\begin{split}
\phi_n(t) 
&=\sqrt{\frac{\sigma'}{\pi}}\sinc[\sigma'(t-t_n)]\theta(t-t_n)\\
&=\sqrt{\frac{\sigma'}{\pi}}\frac{\sin(\sigma' t-n\pi)}{(\sigma' t-n\pi)}
\left[\frac{\sin\frac{\delta}{m}(\sigma' t-n\pi)}
{\frac{\delta}{m}(\sigma' t-n\pi)}\right]^m.
\end{split}
\end{equation}
Clearly, $\phi_n(t)\in\fs{B}^1_{\sigma'_+}$. For fixed $t$, the Helms and
Thomas~\cite{HT1962,J1966} expansion of $p(t+s)$ reads as
\begin{equation}
p(t+s)=\sum_{n\in\field{Z}}\sqrt{\frac{\pi}{\sigma'}}
p\left(t+\frac{n\pi}{\sigma'}\right)\phi_n(s).
\end{equation}
This series converges under much weaker conditions on $p$, however, we would
still restrict ourselves to the case of $p\in\fs{B}^{\nu}_{\sigma'_-}(\equiv\fs{B}^{\nu}_{\sigma})$ with
$\nu=1,2$. For our purpose, it suffices to note that
$\{\phi_n\}_{n\in\field{Z}}$ spans $\fs{B}^{2}_{\sigma'_-}$. 

With this modified basis function, the operator defined by~\eqref{eq:psi-OP},
\begin{equation}
\OP{S}[g](y)\equiv\int^{\infty}_{0}\phi_0(y+s)g(s)ds,
\end{equation}
now defines a bounded linear operator from $\fs{L}^1(\Omega_+)$ to $\fs{L}^1$ with 
\begin{equation}
\|\OP{S}\|_{\fs{L}^1}\leq\int_{\Omega_+}|\phi_0(s)|ds<\|\theta\|_{\fs{L}^2}.
\end{equation}
Moreover, we note that $\OP{S}$ defines a bounded linear operator from
$\fs{L}^1(\Omega_+)$ to $\fs{B}^1_{\sigma'_+}$. Following Jagerman~\cite{J1966}, an 
improved version of the result presented in Prop.~\ref{prop:jagerman-type} is as follows.  
\begin{prop}
For $p$ satisfying an estimate of the form 
\begin{equation}
|p(z)|\leq\frac{C}{(1+|z|)^{k+1}}e^{\sigma'_- |\Im(z)|},\quad z\in\field{C},
\end{equation}
where $k\geq1$ and $g$ satisfying a similar estimate with the index $k'\geq m$, the truncation
error $T_N(t)$ of the partial sums in~\eqref{eq:OP-P-finite} with the basis
functions defined by~\eqref{eq:new-basis} satisfies the
estimate
\begin{equation}
|T_N(t)|\leq
\frac{2(\sigma'/\pi)^{k}\mathcal{E}_k(t)}{(N+1)^{k}\sqrt{1-4^{-k}}}\frac{C'}{(N+1)^{m+1/2}},
\end{equation}
for some constant $C'>0$.
\end{prop}
\begin{proof}
Noting that $g\in\fs{B}_{\sigma'_-}^1$  so that $\theta g\in\fs{B}_{\sigma'}^1$.
Consequently,
\begin{equation*}
\begin{split}
|\hat{g}_N|
&\leq|g(N\pi/\sigma')|
+\sqrt{\frac{\sigma'}{\pi^3}}\left(\frac{\delta\pi}{m}\right)^{-m}N^{-m-1}\int^0_{-\infty}|g(s)|ds\\
&=\bigO{N^{-1-m}},
\end{split}
\end{equation*}
therefore, $|\hat{g}_{\pm N}|=\bigO{N^{-1-m}}$. Now,
\[
\sum_{|n|>N}\frac{1}{N^{2m+2}}\leq2\int_{N+1}^{\infty}\frac{ds}{s^{2m+2}}=\frac{2}{(2m+1)(N+1)^{2m+1}},
\]
so that
\[
\left(\sum_{|n|>N}|\hat{g}_n|^2\right)^{1/2}\leq\frac{C'}{(N+1)^{m+{1}/{2}}}.
\]
for some constants $C'>0$. Other details of the proof are same as that of
Prop.~\ref{prop:jagerman-type}; therefore, we omit it.
\end{proof}
Working from equations~\eqref{eq:discrete-error} and~\eqref{eq:error-estimate2}, if
we assume $\mathcal{M}_N$ to be exact, then the proposition~\ref{eq:convg-WKS}
can be modified as follows:
\begin{prop}
Let $p$ satisfy an estimate of the form 
\begin{equation}
|p(z)|\leq\frac{C}{(1+|z|)^{k+1}}e^{\sigma'_- |\Im(z)|},\quad z\in\field{C},
\end{equation}
where $k\geq m+1/2,\,m\in\field{Z}_+$. Then, for fixed $t\in\field{R}$
and with the basis functions defined by~\eqref{eq:new-basis},
we have
\[
|R^{(N)}_j(t)|=\bigO{N^{-1/2-m}},\quad j=1,2,
\]
where $2N+1$ is the number of basis functions used in the sampling expansion. 
\end{prop}

We conclude this section with a discussion of the quadrature method for
computing the entries of the quadrature matrix
\begin{equation}
\mathcal{Q}_{nl}=\int_{0}^{\infty}\phi_n(s)\phi_l(s)ds,\quad l,n\in\field{Z},
\end{equation}
where $\{\phi_n(s)\}$ are the new basis functions introduced in this section. By a 
straightforward application of the sampling theorem, a simple quadrature rule 
can be worked out as follows:
\begin{equation}
\mathcal{Q}^{(M)}_{nl}=\frac{1}{2\sigma'_+}\sum_{k=-M}^{M}
\phi_n\left(\frac{k\pi}{2\sigma'_+}\right)
\phi_l\left(\frac{k\pi}{2\sigma'_+}\right)
\left[\frac{\pi}{2}-\Si(k\pi)\right],
\quad l,n\in\field{Z}.
\end{equation}
Defining the matrices $\mathcal{D}$ and ${\Phi}$, 
\begin{equation}
\begin{split}
\mathcal{D}_{nk}
&=\frac{1}{2\sigma'_+}\left[\frac{\pi}{2}-\Si(k\pi)\right],
\quad -M\leq m,n\leq M,\\
{\Phi}_{kl}
&=\phi_l\left(\frac{k\pi}{2\sigma'_+}\right),\quad-M\leq k\leq M,\,
-N\leq l\leq N,
\end{split}
\end{equation}
the quadrature matrix truncated to size $N\times N$ can be written as
\begin{equation}
\mathcal{Q}^{(M)}_N=\Phi^{\tp}\mathcal{D}\Phi.
\end{equation}
The quadrature errors in the present case can be treated in exactly the same
manner as before and given that the new basis functions have better decay
properties, the results contained in the lemma~\ref{lemma:Q-errors} remain valid.
%%%%%%%%%%%%%%%%%%%%%%%%%%%%%%%%%%%%%%%%%%%%%%%%%%%%%%%%%%%%%%%%%%%%%%%%%%%%%%%%%%%%%%%%%%%%%%%
%%%%%%%%%%%%%%%%%%%%%%%%%%%%%%%%%%%%%%%%%%%%%%%%%%%%%%%%%%%%%%%%%%%%%%%%%%%%%%%%%%%%%%%%%%%%%%%
%%%%%%%%%%%%%%%%%%%%%%%%%%%%%%%%%%%%%%%%%%%%%%%%%%%%%%%%%%%%%%%%%%%%%%%%%%%%%%%%%%%%%%%%%%%%%%%
%%%%%%%%%%%%%%%%%%%%%%%%%%%%%%%%%%%%%%%%%%%%%%%%%%%%%%%%%%%%%%%%%%%%%%%%%%%%%%%%%%%%%%%%%%%%%%%
\section{Numerical and Algorithmic Aspects}\label{sec:num}
Based on the analysis presented in the earlier sections, the discrete system to be solved has 
the form 
\begin{equation}
\begin{pmatrix}
I & -\mathcal{Q}\mathcal{P}\\
-\kappa\mathcal{Q}\mathcal{P}^{*} & I
\end{pmatrix}
\begin{pmatrix}
\vv{v}_1\\
\vv{v}_2
\end{pmatrix}
=
\begin{pmatrix}
\mathcal{Q}{\vv{p}}\\
\vv{0}
\end{pmatrix},
\end{equation}
where $\mathcal{Q}\in\field{R}^{N\times N}$ is a symmetric positive-definite
matrix, $\mathcal{P}\in\field{C}^{N\times N}$ is a Hankel matrix and
$\kappa\in\{+1,-1\}$. We symmetrize the linear
system by introducing $\vs{v}_j=\mathcal{Q}^{1/2}\vs{u}_j,\,j=1,2$ so
that
\begin{equation*}
\begin{pmatrix}
I & -\mathcal{Q}^{1/2}\mathcal{P}\mathcal{Q}^{1/2}\\
-\kappa\mathcal{Q}^{1/2}\mathcal{P}^*\mathcal{Q}^{1/2} & I
\end{pmatrix}
\begin{pmatrix}
\vv{u}_1\\
\vv{u}_2
\end{pmatrix}
=
\begin{pmatrix}
\mathcal{Q}^{1/2}{\vv{p}}\\
\vv{0}
\end{pmatrix}.
\end{equation*}
Putting $\mathcal{G}=\mathcal{Q}^{1/2}\mathcal{P}\mathcal{Q}^{1/2}$, we have
\begin{equation}
\begin{pmatrix}
                 I & -\mathcal{G}\\
-\kappa\mathcal{G}^{\dagger} &  I
\end{pmatrix}
\begin{pmatrix}
\vv{u}_1\\
\vv{u}_2
\end{pmatrix}
=
\begin{pmatrix}
\mathcal{Q}^{1/2}{\vv{p}}\\
\vv{0}
\end{pmatrix},
\end{equation}
which reduces to
\begin{equation}\label{eq:solve-by-CG}
(I-\kappa\mathcal{G}^{\dagger}\mathcal{G})\vs{u}_1
=\mathcal{Q}^{1/2}{\vv{p}},
\end{equation}
which is numerically well conditioned on account of the positive definite nature of
$(I-\kappa\mathcal{G}^{\dagger}\mathcal{G})$ (for $\kappa=+1$, we assume
$\|\mathcal{G}\|_s<1$). This linear system must be solved 
in order to compute $q(0)$. In order to obtain
$q(nh)$, where $h\,(>0)$ is the step size and $n\in\field{Z}$, we must translate
$p(\tau)$ by $2nh$, and, compute $\mathcal{P}$ and $\vv{p}$. If a direct
method of solving a system of linear equation is used, then the complexity of
computation per sample of the potential works out to be $\bigO{N^3}$ (excluding
the cost of computing $\mathcal{P}$ and $\vv{p}$). However, it is also possible
to use an iterative method such as the conjugate-gradient (CG) method to solve the
linear system leading to a complexity of $\bigO{N_{\text{iter.}}N^2}$. Note that
the initial seed for such a procedure can be obtained by using a direct solver at
any fixed $n$, say, $n=0$. The CG-iterations in the subsequent step can be 
seeded by using the solution obtained in the last step. We choose the
threshold for convergence of the CG iterations to be $10^{-12}$ unless
otherwise stated.

The discussion above is valid for both of the methods proposed in the earlier sections. For 
the sake of convenience, let us label the methods by the basis functions used in 
their respective sampling expansions. The first method uses the 
Whittaker-Kotelnikov-Shannon (WKS) sampling series which consists of translates 
of the sinc function. The basis is completely
determined by the bandlimiting parameter $\sigma$; therefore, we label this
method by ``WKS$_{\sigma}$''. The second method uses the Helms and Thomas~\cite{HT1962,J1966}
version of the sampling series. The basis in this case is completely determined
by three parameters: the guard-band parameter $\delta$, the index $m$ of the convergence
accelerating function and the bandlimiting parameter $\sigma'$ given by
\begin{equation}
\sigma'=\sigma/(1-\delta),\quad\delta<1.
\end{equation}
Therefore, we
label this method by ``HT$_{\sigma'}^{(m,\delta)}$''. In the formalism adopted
in Sec.~\ref{sec:HT}, $\sigma'_-=\sigma$. Here, we restrict ourselves 
to the choice $m=4$ and $\delta=0.4$. We may
often drop the subscripts and superscripts in these labels for the sake of
brevity if additional information about the underlying basis is not relevant.
Finally, let us also note that the 
basis functions involved in each of the methods are symmetrically translated
copies of the zero index basis function, i.e., setting the number of basis
functions to be $N=2N_{\text{shift}}+1$, 
the maximum translation about the origin is given by 
$|t_{\text{shift}}|=(\pi/\sigma)N_{\text{shift}}$ and 
$|t_{\text{shift}}|=(\pi/\sigma')N_{\text{shift}}$ for WKS and HT, respectively. 

\begin{figure*}[!htb]
\centering
\includegraphics[scale=1]{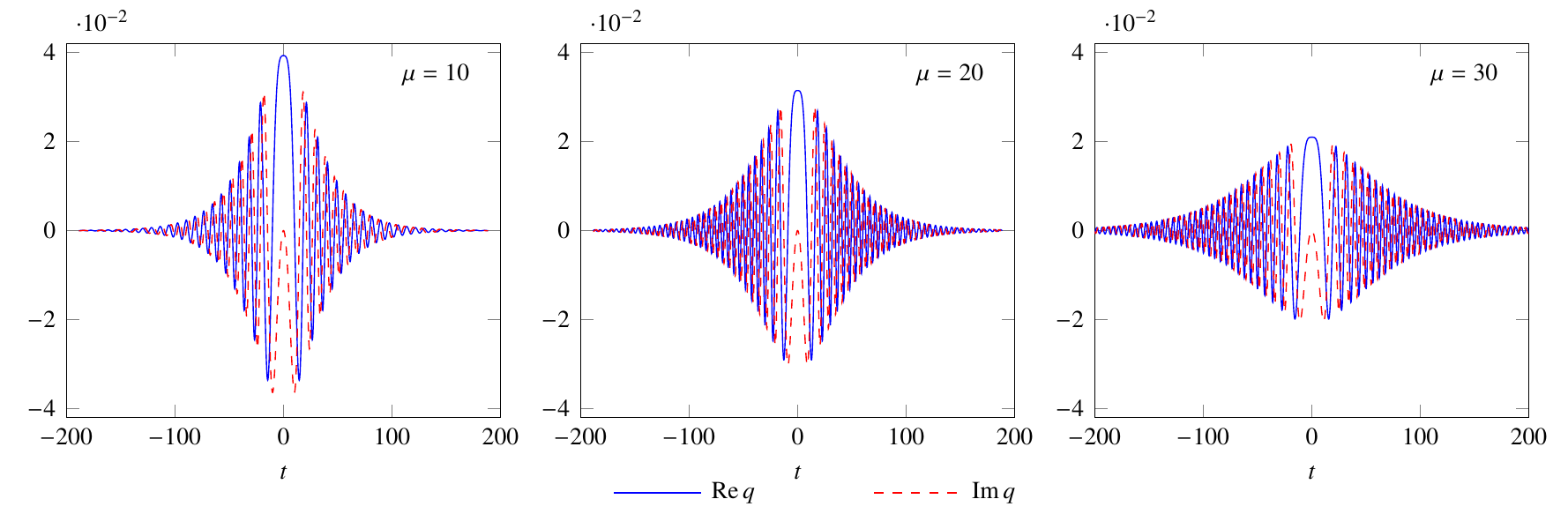}
\caption{\label{fig:sig_csech_exact}The figure shows the 
chirped secant-hyperbolic potentials with chirp parameter given by $\mu=10,20,30,$ and
the corresponding scale parameter $a=80/\pi,100/\pi,150/\pi$, respectively.}
\end{figure*}

\begin{figure*}[!thb]
\centering
\includegraphics[scale=1]{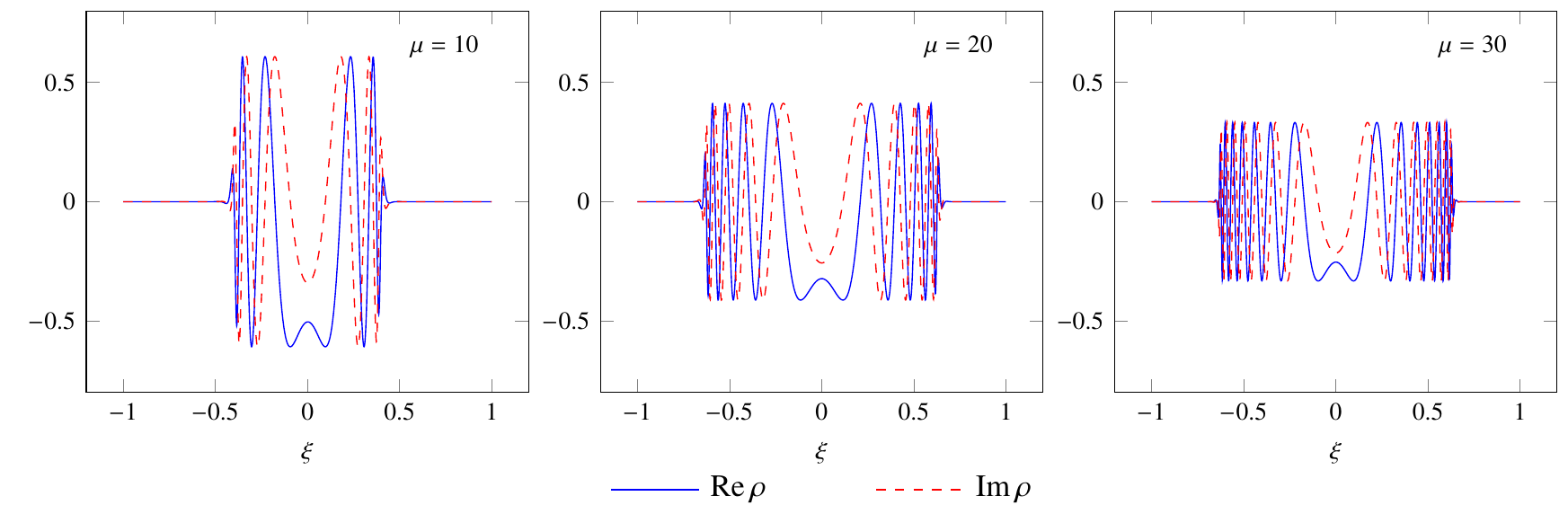}
\caption{\label{fig:csech_rho}The figure shows the reflection coefficient of 
chirped secant-hyperbolic potentials with chirp parameter given by $\mu=10,20,30,$ and
the corresponding scale parameter $a=80/\pi,100/\pi,150/\pi$, respectively.}
\end{figure*}

\begin{figure*}[!htbp]
\centering
\includegraphics[scale=1]{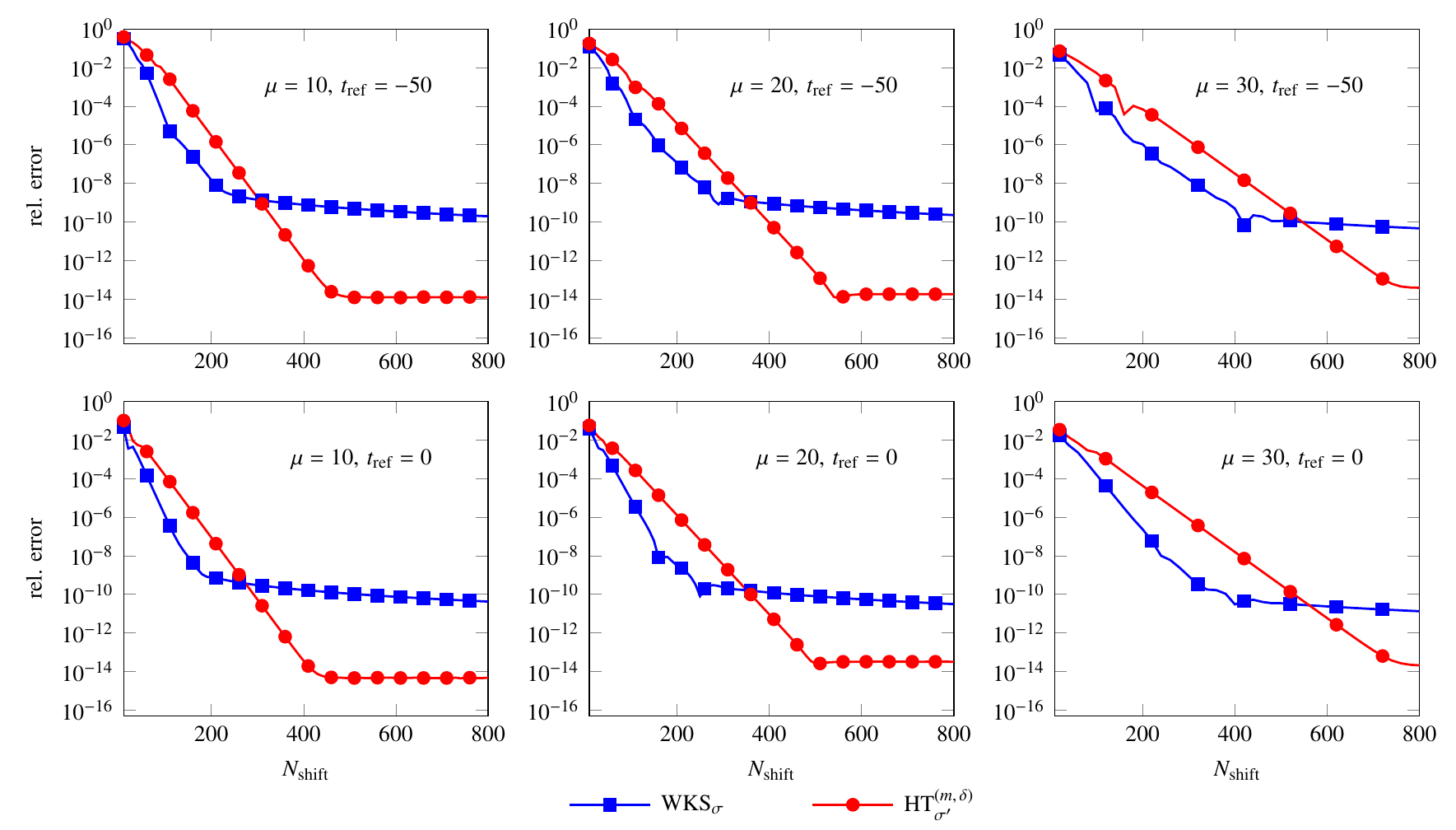}
\caption{\label{fig:convg_csech}The figure shows the convergence analysis of the
algorithms labelled WKS$_{\sigma}$ and HT$_{\sigma'}^{(m,\delta)}$ for the chirped secant-hyperbolic
potential with $\mu=10,20,30$ at the positions
$t_{\text{ref.}}=-50$ (top row) and $t_{\text{ref.}}=0$ (bottom row). The basis
functions used in either of these methods are translated symmetrically about the
origin and the number of basis functions is $N=2N_{\text{shift}}+1$. In
order to make the potential effectively bandlimited, the scale parameter is chosen to be $a=80/\pi,
100/\pi, 150/\pi$, for $\mu=10,20,30$, respectively. For the case $\mu=30$,
the number of LGL quadrature nodes is $N_{\text{quad}}=2000$ while for the rest
$N_{\text{quad}}=1000$. Before the 
error plateaus, the slope of the error curve is consistent with an 
exponential rate of convergence.}
\end{figure*}

\begin{figure*}[!htb]
\centering
\includegraphics[scale=1]{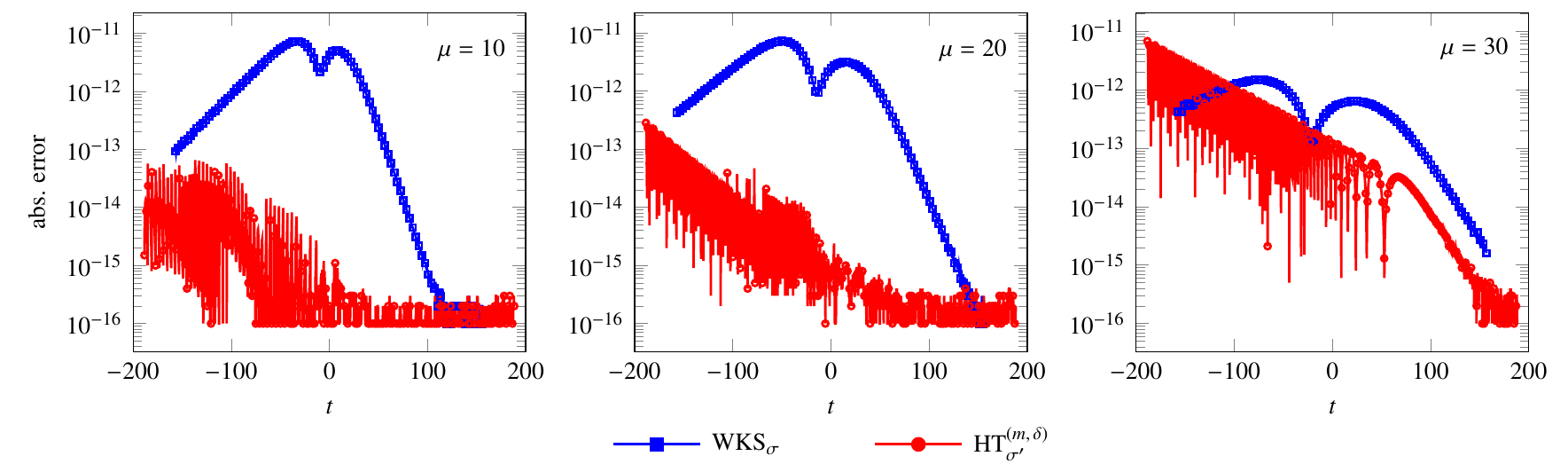}
\caption{\label{fig:sig_csech} The figure shows the absolute error in the
computed potential for the case of chirped secant-hyperbolic potential with
$\mu=10,20,30$. The computations are carried out with $N_{\text{shift}}=500$ for
the first two cases while the same for the last one is $N_{\text{shift}}=600$.
The number of LGL quadrature nodes for the first two cases is $N_{\text{quad}}=1000$
while for the last one $N_{\text{quad}}=2000$.}
\end{figure*}
 
Now, the estimate of computational complexity provided above excludes the
cost of computing the samples of $p(\tau)$ needed to compute $\mathcal{P}$ 
and $\vv{p}$ at each of the steps\footnote{By steps we mean
progressive translation of $p(\tau)$.}.  In view of the fact that function evaluations 
are, in general, expensive computationally, our algorithm should ensure that they
are used optimally. Given that $\mathcal{P}$ is a Hankel matrix, it
suffices to compute the first column and the last row which amounts to $2N-1$
evaluations of $p(\tau)$. The vector $\vv{p}$ is related to
the column vector $(\mathcal{P}_{k,0})$; therefore, it does not require additional evaluations of
the impulse response. Let the translation of $p(\tau)$, for the method WKS, be in the steps of $2h$
determined by
\begin{equation}\label{eq:n-os}
n_{\text{os}}h={\pi}/{\sigma},\quad h>0,\,n_{\text{os}}\geq1,
\end{equation}
where $n_{\text{os}}$ is referred to as the
\emph{over-sampling factor}. Consequently, the nodes over which one needs to sample
$p(\tau)$ are of the form $\tau_j=jh,\,j\in\field{Z}$. If the potential is supposed to
be determined over the grid $t_k=kh$ where $k=-K,-K+1,\ldots,K-1,K,$ and the number
of basis functions is $N$, then the impulse response must be sampled at the grid
points 
\begin{equation}
\left\{\begin{aligned}
&\tau_j=jh,\quad j=-\ovl{N},-\ovl{N}+1,\ldots,\ovl{N}-1,\ovl{N},\\
&\ovl{N}=n_{\text{os}}N+2K.
\end{aligned}\right.
\end{equation}
For the method HT, $\sigma$ is replaced by $\sigma'$ in~\eqref{eq:n-os} while
all the other aspects remain the same. For the examples considered in this
section, we set $\sigma=1$ and $n_{\text{os}}=10$ unless otherwise stated. 

With regard to the input required for the aforementioned methods, let us note that the nonlinear 
impulse response may not be available in a closed form. In fact, the inverse NFT is defined to take
the reflection coefficient $\rho(\xi)$ as input. The samples of the impulse response can
then be computed using the FFT algorithm with an appropriately large
over-sampling factor. Alternatively, if extremely high accuracy is demanded, we
may use methods that are specially designed for highly oscillatory integrals
such as the Fourier integral~\cite[Chap.~3]{DR1984}. One such method, 
attributed to Bakhvalov and Vasil'eva~\cite{BV1968}, is described in the 
Appendix~\ref{app:NIR-Bessel} where Legendre-Gauss-Lobbato (LGL) quadrature is used to
obtain the nonlinear impulse response in terms of the spherical Bessel
functions. In our tests, we have employed the latter method with the number of
LGL nodes set to $N_{\text{quad.}}=1000$ unless otherwise stated.

\begin{figure*}[!thb]
\centering
\includegraphics[scale=1]{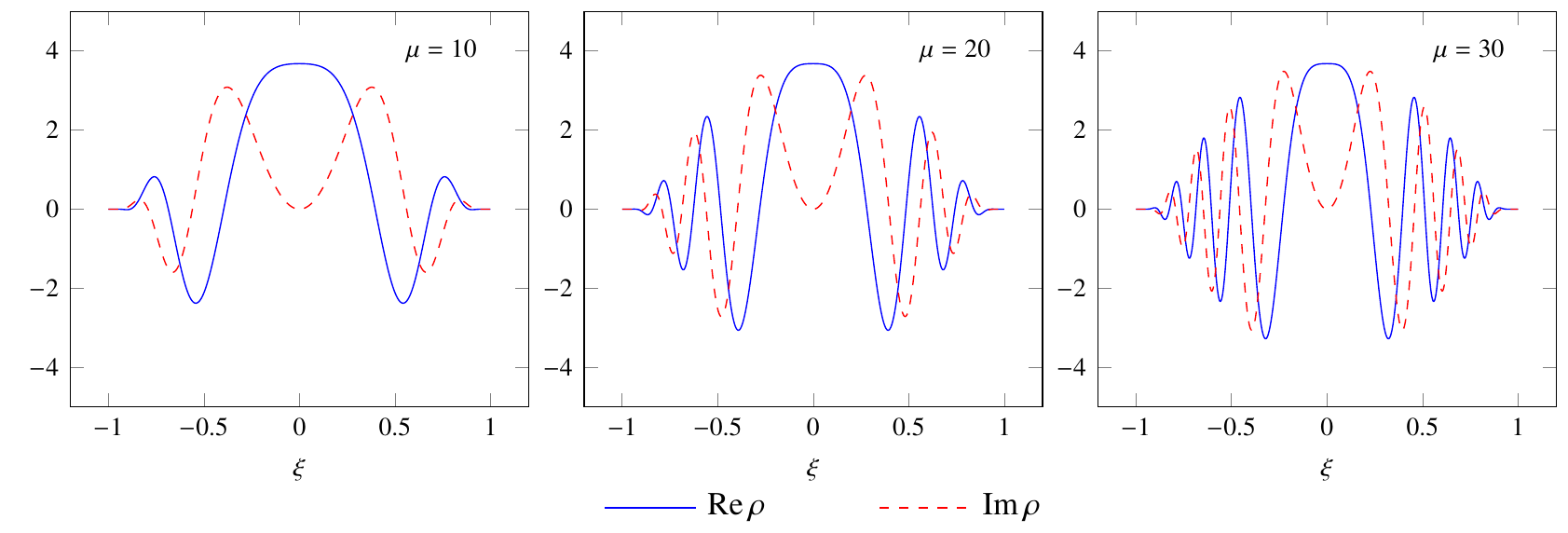}
\caption{\label{fig:bump_rho}The figure shows a chirped bump function defined
by~\eqref{eq:chirped-bump} where $A_0=10$, $\sigma=1$ and the chirp parameter 
$\mu\in\{10,20,30\}$.}
\end{figure*}

\begin{figure*}[!thb]
\centering
\includegraphics[scale=1]{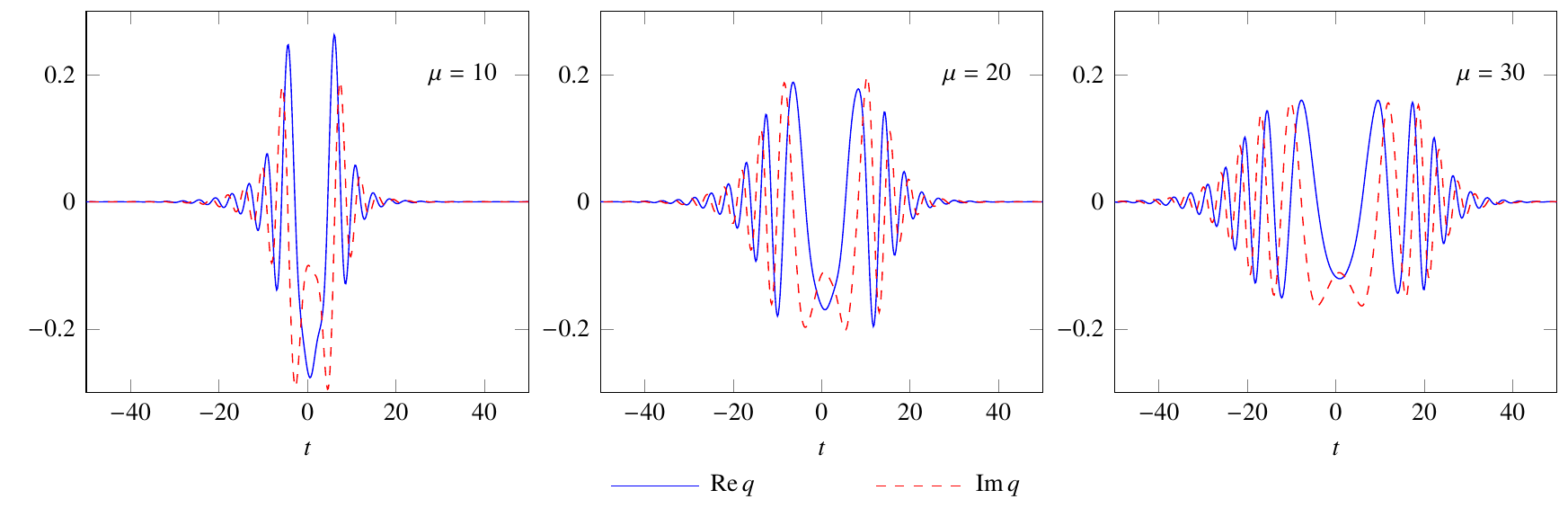}
\caption{\label{fig:sig_bump_ref}The figure shows the numerically computed scattering potential
corresponding to the chirped bump function as reflection coefficient depicted in
Fig.~\ref{fig:bump_rho}.}
\end{figure*}

\begin{figure*}[!thb]
\centering
\includegraphics[scale=1]{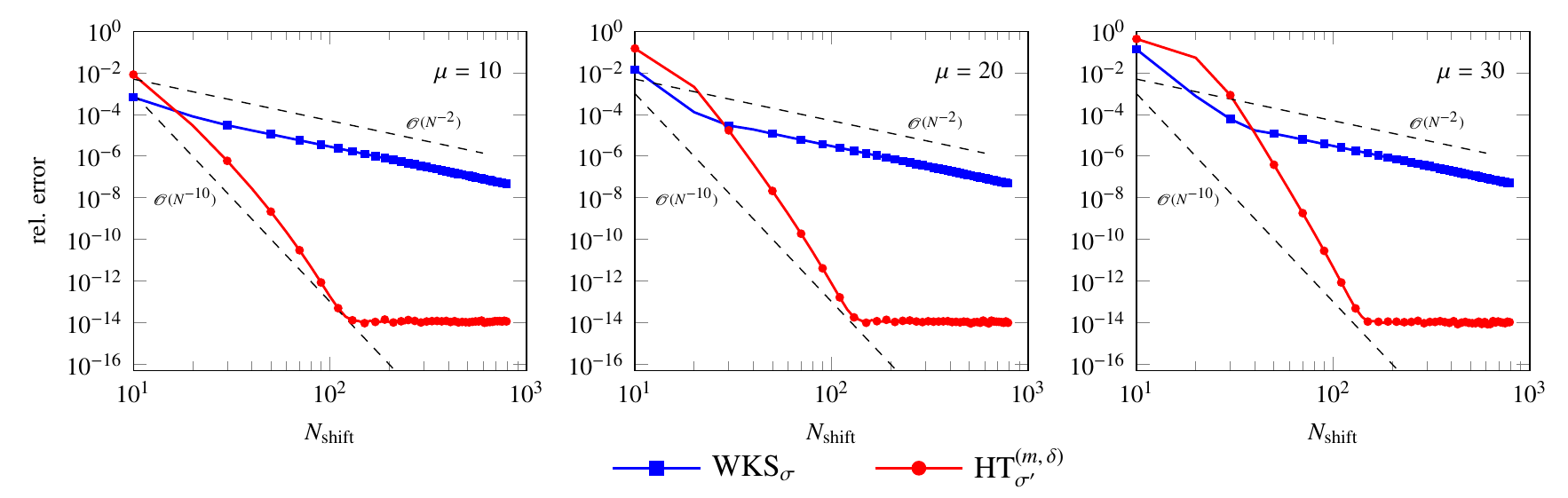}
\caption{\label{fig:convg_bump}The figure shows the convergence analysis of the
algorithms WKS$_{\sigma}$ and HT$^{(m,\,\delta)}_{\sigma'}$ for the chirped bump
function as reflection coefficient defined by~\eqref{eq:chirped-bump} with 
$A_0=10$ and $\mu\in\{10,20,30\}$. The error is quantified by~\eqref{eq:err-ref}
(with the reference solution computed using HT with $N_{\text{shift}}=2000$ and
$N_{\text{quad}}=6000$) and the number of basis functions, symmetrically translated about the origin, is
$N=2N_{\text{shift}}+1$. The slope of the
error curve for the method WKS is unambiguously consistent with a second order
of convergence. For the method HT, an algebraic rate of convergence 
better that $\bigO{N^{-10}}$ can be obtained using a linear fit before the error
plateaus.}
\end{figure*}

\begin{figure*}[!thb]
\centering
\includegraphics[scale=1]{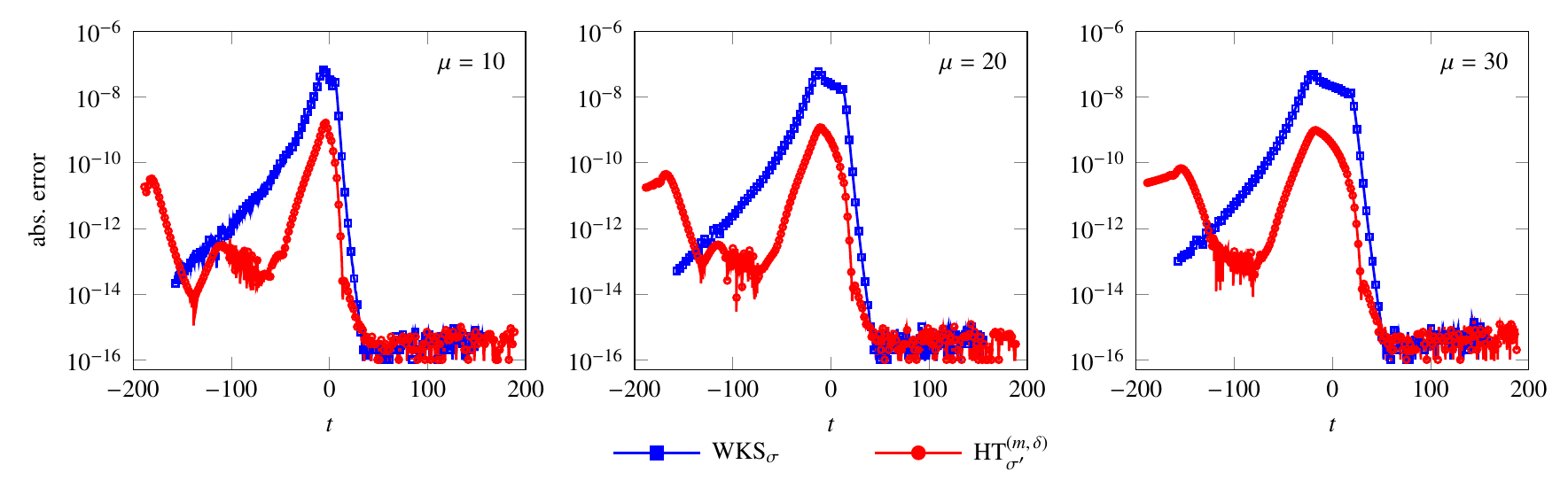}
\caption{\label{fig:sig_bump}The figure shows the error in the numerically computed
scattering potential corresponding to the chirped bump
function as reflection coefficient defined by~\eqref{eq:chirped-bump} with 
$A_0=10$ and $\mu\in\{10,20,30\}$. The error for the algorithms 
WKS$_{\sigma}$ and HT$^{(m,\,\delta)}_{\sigma'}$ is quantified by~\eqref{eq:err-ref} where the reference
solution is computed using the fast inverse NFT algorithm reported
in~\cite{V2018BL} with $2^{21}$ number of samples and the step-size is $2^{10}$
times smaller than that used in the algorithm being tested.}
\end{figure*}

\begin{figure*}[!thb]
\centering
\includegraphics[scale=1]{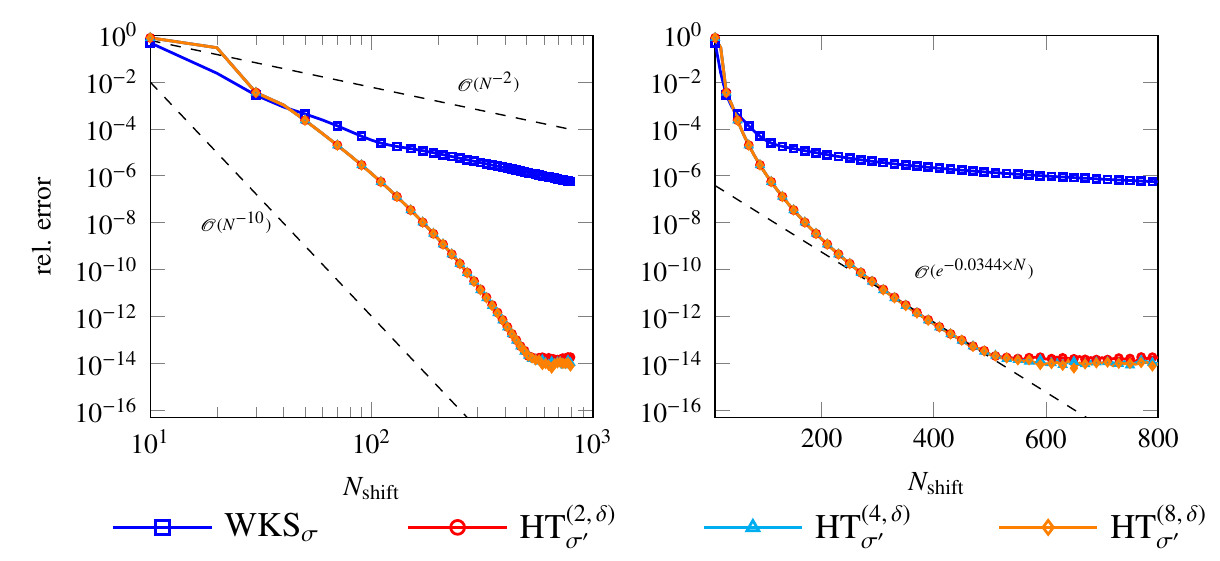}
\caption{\label{fig:comp_bump}The figure shows the convergence analysis of the
algorithms WKS$_{\sigma}$ and HT$^{(m,\,\delta)}_{\sigma'},\,m\in\{2,4,8\},$ for the chirped bump
function as reflection coefficient defined by~\eqref{eq:chirped-bump} with 
$A_0=10$, $n=5$ and $\mu=30$. }
\end{figure*}

\begin{figure*}[!th]
\centering
\includegraphics[trim=0.5cm 1.5cm 1cm 1cm, clip]{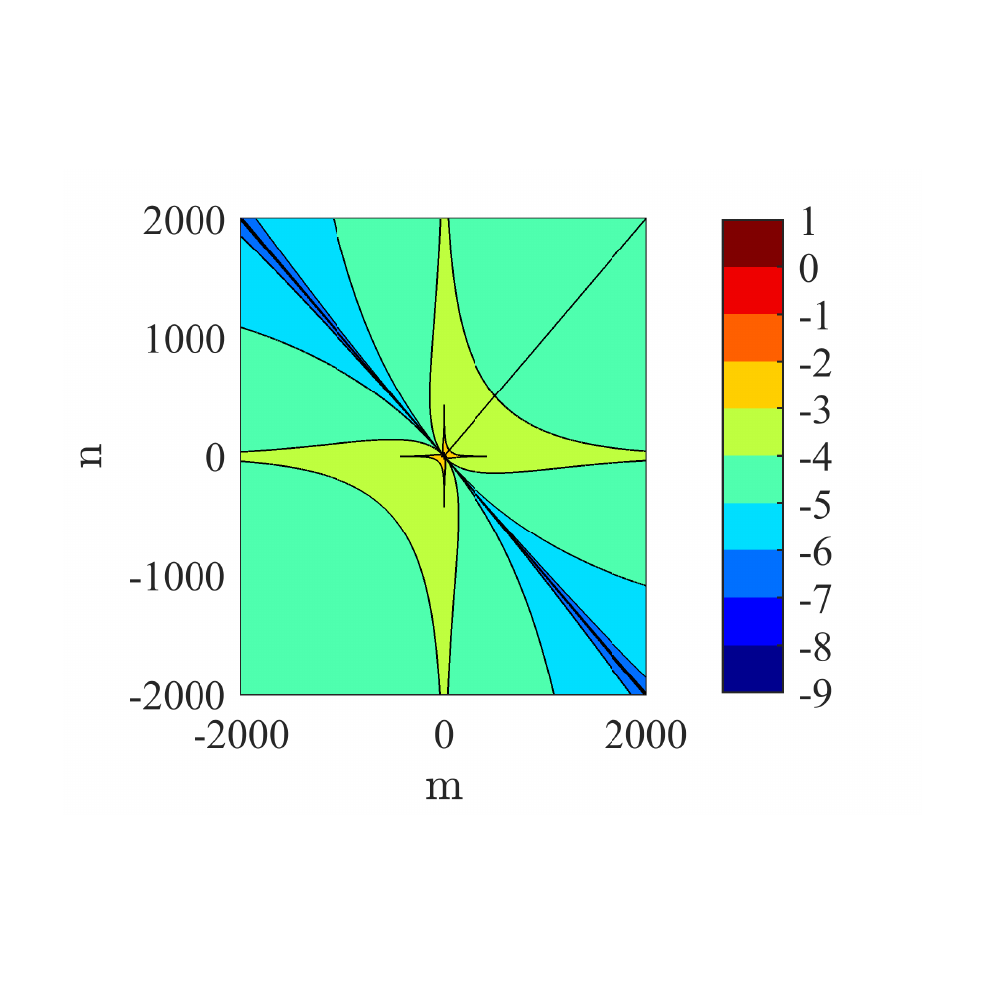}
\includegraphics[trim=0.5cm 1.5cm 1cm 1cm, clip]{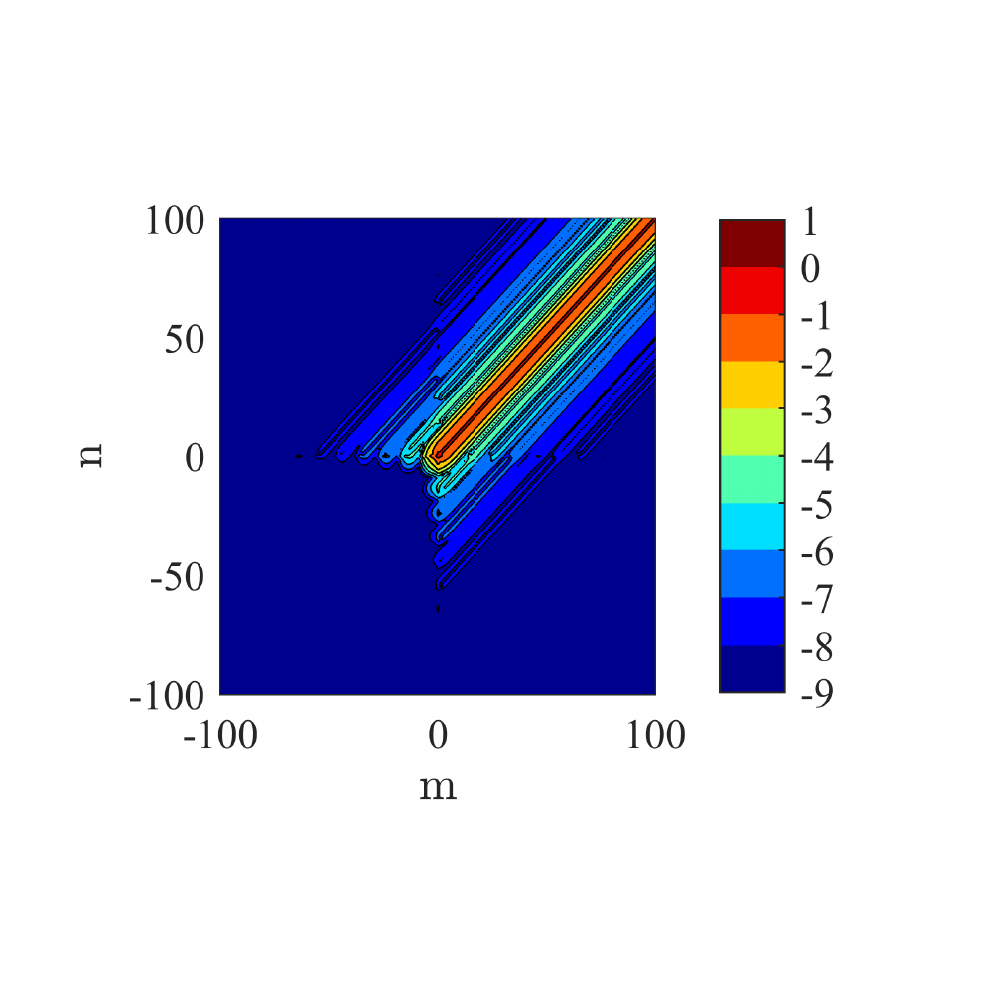}
\caption{\label{fig:qmat} The figure show the contour plot of
$\log_{10}|\mathcal{Q}_{mn}|$ for the method WKS (left) and the method HT$_{\sigma}^{(m,\,\delta)}$.}
\end{figure*}

Now we turn to to the error analysis of the proposed methods. In these tests, we 
restrict ourselves to the case $\kappa=-1$. For the purpose of convergence analysis, we 
choose the chirped secant-hyperbolic potential~\cite{TVZ2004}:
\begin{equation}
q(t) = A_0\frac{\exp[-2i\mu A_0\log(\cosh t)]}{\cosh(t)},
\end{equation}
which is not a nonlinearly bandlimited signal, however, it can be
considered effectively bandlimited\footnote{For a reflection coefficient which
is not compactly supported, if
$|\rho(\xi)|\leq{C}{(1+|\xi|)^{-\nu-1}}$ for $\xi\in\field{R}$ and some $\nu>0$, then
\begin{equation*}
\left|p(\tau)-\frac{1}{2\pi}\int_{-\sigma}^{\sigma}\rho(\xi)e^{i\xi\tau}d\xi\right|
\leq\frac{C}{\nu\pi(1+\sigma)^{\nu}},\quad \tau\in\field{R}.
\end{equation*}
Therefore, by choosing $\sigma$ large enough one can consider $\rho(\xi)$ as
effectively bandlimited.}. We assume that $\mu \geq 1$ so that the discrete spectrum 
is empty. The reflection coefficient is given by
\begin{equation}
\begin{split}
\rho(\xi)&= 
-\frac{A_0e^{-2i\mu A_0(\log2)}}
{\Gamma\left(1-\frac{iA_0\omega}{2\mu}\right)\Gamma\left(1-\frac{2iA_0\mu}{\omega}\right)}\cdot
\frac{\Gamma\left(\frac{1}{2}+i\xi-iA_0\mu\right)}
{\Gamma\left(\frac{1}{2}-i\xi+iA_0\mu\right)}\times\\
&\qquad\Gamma\left(\frac{1}{2}-i\xi+i\lambda\right)
\Gamma\left(\frac{1}{2}-i\xi-i\lambda\right),\\
\lambda &=A_0\mu\sqrt{1-\mu^{-2}},\quad\omega = \frac{2}{1 + \sqrt{1-\mu^{-2}}}.
\end{split}
\end{equation}
As $|\xi|\rightarrow\infty$, the reflection coefficient decays as
$\text{const.}\times e^{-\pi|\xi|}$. We set $A_0=1$ and let
$\mu\in\{10,20,30\}$. The potential corresponding to these choices of the
parameters is shown Fig.~\ref{fig:sig_csech_exact} with the corresponding
reflection coefficient shown in Fig.~\ref{fig:csech_rho}. In
the tests, we take the input as $\rho(a\xi)$ where $a$ is large enough so that
$\sigma=1$, effectively. Given that the scattering potential at any point on the chosen grid can be 
computed independently of other points,
it suffices to test the convergence of the methods at any arbitrary point, say,
$t_{\text{ref.}}$. We then quantify the error by
\begin{equation}\label{eq:err-true}
e_{\text{rel.}} = \frac{|q(t_{\text{ref.}})-q^{(\text{num.})}(t_{\text{ref.}})|}{|q(t_{\text{ref.}})|}.
\end{equation}
Note that, for the determination of the rate of convergence, we resort to 
a direct solver for the linear system involved in order to avoid all possible
sources of error.

The results of the convergence analysis is shown in Fig.~\ref{fig:convg_csech}.
It turns out that the rate of convergence in these examples is superior than
what is theoretically predicted. Both the methods exhibit exponential rate of
convergence before plateauing of the error curves takes place. Note that the
best accuracy achievable is remarkably close to the machine precision. Next, we may
also want to examine the pointwise error in the computed potential over a set of
grid points in order to
ascertain if the CG iteration converge to the right solution. This is tested in
Fig.~\ref{fig:csech_rho} which is consistent with the error levels reported in
the convergence analysis.

The next example is of a compactly supported reflection coefficient, 
the chirped  ``bump function'':
\begin{equation}\label{eq:chirped-bump}
\rho(\xi) =
A_0\exp\left[-\frac{1}{1-\left(\frac{\xi}{\sigma}\right)^{2n}}
+i\mu\left(\frac{\xi}{\sigma}\right)^2\right]\chi_{[-\sigma,\sigma]}.
\end{equation}
We set $\sigma=1$, $A_0=10$, $n=1$ and let $\mu\in\{10,20,30\}$. The potential corresponding 
to these choices of the parameters is shown Fig.~\ref{fig:sig_bump_ref} with the corresponding
reflection coefficient shown in Fig.~\ref{fig:bump_rho}. In the absence of a closed form 
solution of the inverse scattering problem, we choose to quantify the error by
\begin{equation}\label{eq:err-ref}
e_{\text{ref.}} = \frac{|q^{(\text{ref.})}(t_{\text{ref.}})-q^{(\text{num.})}(t_{\text{ref.}})|}
{|q^{(\text{ref.})}(t_{\text{ref.}})|},
\end{equation}
where $q^{(\text{ref.})}$ is the solution obtained using the method
HT$_{\sigma}^{(m,\,\delta)}$ with $N_{\text{shift}} = 2000$ and
$N_{\text{quad.}}=6000$. The results of the error analysis in this example 
must be interpreted with caution because $e_{\text{ref.}}$ is not the true
numerical error. The results of this numerical experiment is shown in 
Fig.~\ref{fig:convg_bump} where the method WKS shows an algebraic rate of convergence
(which also turns out to be superior than what was predicted). However, the convergence
behavior of HT is does not immediately confirm an algebraic rate because it
seems to change to an exponential rate. To clarify this, let us compare the
methods HT$^{(m,\,\delta)}_{\sigma'},\,m\in\{2,4,8\},$ for the chirped bump
function as reflection coefficient defined by~\eqref{eq:chirped-bump} with 
$A_0=10$, $n=5$ and $\mu=30$. The results are shown in 
Fig.~\ref{fig:comp_bump} where the plot on the right seems to confirm the earlier 
observation that convergence behavior might be exponential. Based on these
observation it reasonable to expect that the HT method exhibits exponential 
convergence for Schwartz class impulse response. A theoretical justification
for these observation is not available yet and we hope to address this in
the future.

The pointwise
error over a set of grid points is shown in Fig.~\ref{fig:sig_bump}. The
reference solution in this case is computed using the fast inverse NFT reported
in~\cite{V2018BL} with $2^{21}$ number of samples and the step-size is
$2^{-10}$-th of that used in the WKS or the HT method. The degree of agreement with the
reference solution is consistent with the convergence behavior determined earlier.

\begin{figure*}[!tbh]
\centering
\includegraphics[scale=1]{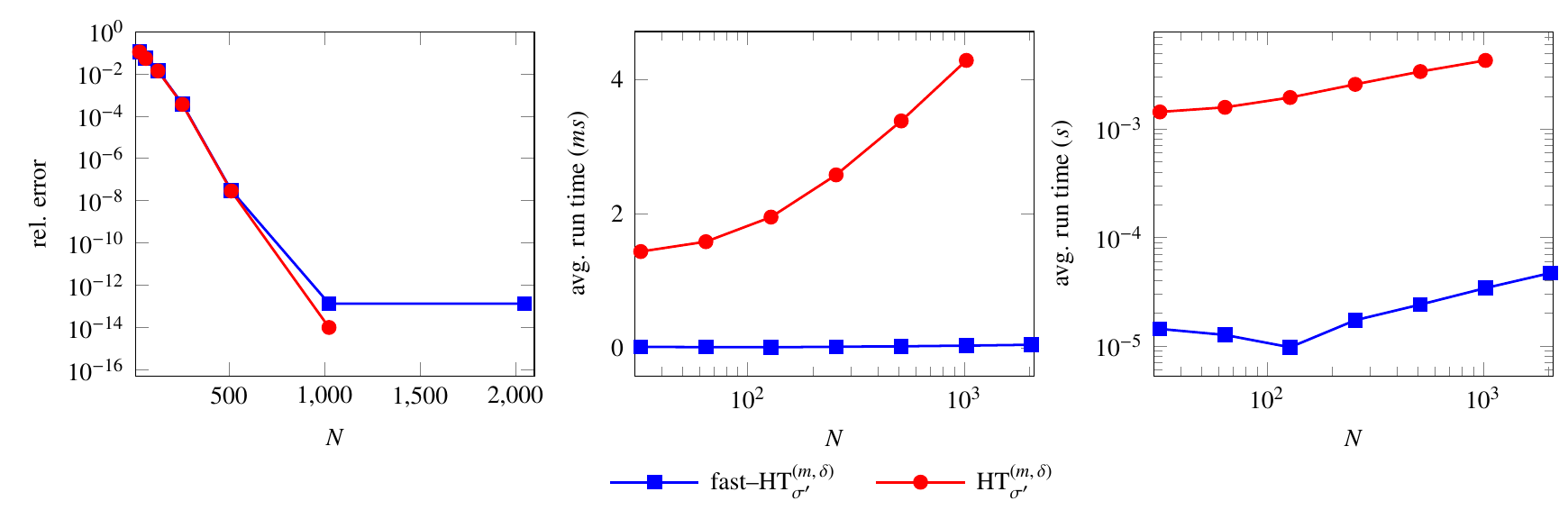}
\caption{\label{fig:comp_fnbl}The figure shows a comparison of convergence and
run-time behavior of the fast variant of HT$_{\sigma}^{(m,\,\delta)}$ with 
HT$_{\sigma}^{(m,\,\delta)}$. The example chosen for this experiment is the
chirped secant-hyperbolic potential with $\mu=10$. The tolerance for selecting
the dominant diagonals of the quadrature matrix in the fast algorithm is chosen to
be $10^{-12}$.}
\end{figure*}

\subsection{Fast Solver using a Sparse Approximation}
In this section, we would like to discuss how a fast variant of the method 
HT$_{\sigma}^{(m,\,\delta)}$ can be obtained by introducing a tolerance
$\epsilon$ to approximate its dense quadrature matrix with a sparse banded matrix.
This idea is motivated by the contour plot of the quadrature matrix 
in Fig.~\ref{fig:qmat}. Clearly, the quadrature 
matrix for the method HT exhibits an effectively banded structure compared to that of
WKS. The nature of the contour map of $\mathcal{Q}$
for HT can be easily understood as follows: Recalling
\begin{equation}
\mathcal{Q}_{nl}=\int_{0}^{\infty}\phi_n(s)\phi_l(s)ds,\quad l,n\in\field{Z},
\end{equation}
and letting $m,n<0$, we have
\begin{equation}
|\mathcal{Q}_{nl}|\leq\left(\frac{m}{\pi\delta}\right)^{2m}|l|^{-m}|n|^{-m}.
\end{equation}
If only $n<0$, then
\begin{equation}
|\mathcal{Q}_{nl}|\leq\left(\frac{m}{\pi\delta}\right)^{m}|n|^{-m}\|\phi_{l}\|_{\fs{L}^{2}}.
\end{equation}
Appealing to the symmetric nature of $\mathcal{Q}$, similar conclusion holds for
$l<0$. Therefore, the dense part of the matrix $\mathcal{Q}$ falls in the quadrant
where $n,l>0$. Consider $n,l>0$ and $n\neq l$. Then, without loss of generality, we can assume 
$n<l$ so that 
\begin{equation}
\begin{split}
|\mathcal{Q}_{nl}| & \leq\int_{0}^{(n+l)\frac{\pi}{2\sigma}}|\phi_n(s)\phi_l(s)|ds
+\int_{(n+l)\frac{\pi}{2\sigma}}^{\infty}|\phi_n(s)\phi_l(s)|ds\\
& \leq
\left(\frac{2m}{\pi\delta}\right)^{m}\frac{\|\phi_n\|_{\fs{L}^2}+\|\phi_l\|_{\fs{L}^2}}{|n-l|^{m}}
\le2\left(\frac{2m}{\pi\delta}\right)^{m}\frac{1}{|n-l|^{m}}.
\end{split}
\end{equation}
For $\epsilon>0$,
\begin{equation*}
%\begin{split}
%&2\left(\frac{2m}{\pi\delta}\right)^{m}\frac{1}{|n-l|^{m}}\leq\epsilon,\\
|n-l|\geq \left(\frac{2m}{\pi\delta}\right)\left(\frac{2}{\epsilon}\right)^{1/m},
%\end{split}
\end{equation*}
ensures that $|\mathcal{Q}_{nl}|\leq\epsilon$. Based on the preceding
inequalities, one can define the number of dominant diagonals, say,
$2N_{\text{band}}$ by
\begin{equation}
N_{\text{band}}(\epsilon)=\left[\left(\frac{2m}{\pi\delta}\right)
\left(\frac{2}{\epsilon}\right)^{1/m}\right] + 1,
\end{equation}
where $[x]$ denotes the integral part of $x\in\field{R}_+$. While this estimate
is important as it sets the upper bound\footnote{This bound can facilitate a
search based algorithm to look for more precise value of the number of dominant
diagonals. We leave these issues for future research.}, it is not so useful in practice because 
it greatly overestimates the number of dominant diagonals. Given that the quadrature matrix 
needs to be computed only once, it is rather easy to check the entries directly and determine the
sparsity of this matrix. We choose to set this tolerance to be
$\epsilon=10^{-12}$. Let $\mathcal{Q}_{\epsilon}$ and
$\mathcal{Q}^{1/2}_{\epsilon}$ denote the banded matrices derived from the dense
matrices $\mathcal{Q}$ and $\mathcal{Q}^{1/2}$, then the linear system
in~\eqref{eq:solve-by-CG} can be approximated by
\begin{equation}
(I-\kappa\mathcal{Q}^{1/2}_{\epsilon}\mathcal{P}^{\dagger}
\mathcal{Q}_{\epsilon}\mathcal{P}\mathcal{Q}^{1/2}_{\epsilon})\vs{u}_{1,\epsilon}
=\mathcal{Q}^{1/2}_{\epsilon}{\vv{p}},
\end{equation}
where $\vs{u}_{1,\epsilon}$ approximates $\vs{u}_{1}$. Let us now estimate the 
cost of one CG iteration if the matrix-vector multiplications
involved are carried out in a cascaded fashion. The cost of multiplying 
$\mathcal{Q}_{\epsilon}$ or $\mathcal{Q}^{1/2}_{\epsilon}$ with a vector is $\bigO{N_{\text{band}}N}$,
the cost of multiplying $\mathcal{P}$ or $\mathcal{P}^{\dagger}$ with a vector is
$\bigO{N\log N}$ (where we exploit the fact that they are Hankel matrices). Therefore
the total cost of one CG iteration is $\bigO{N\log N}+\bigO{N_{\text{band}}N}$.
In the asymptotic limit $\log N\gg N_{\text{band}}$ so that the cost works out to be 
$\bigO{N\log N}$. Therefore, the total cost per sample of the scattering
potential works out to be $\bigO{N_{\text{iter.}}N\log N}$. Finally, let us
observe that the approximation introduced above adds an error of
$\bigO{N\epsilon}$ to the original error estimates.

The fast method obtained above can be tested against the original method to 
determine its convergence and run-time behavior. The results of the
numerical experiment with the chirped secant-hyperbolic profile ($\mu=10$) is
shown in Fig.~\ref{fig:comp_fnbl}. Here the average run-time is the run-time per
sample averaged over the number of basis functions
$N\in\{2^5,\ldots,2^{11}\}$. Note that the
improvement in the complexity comes at a price of accuracy as
evidenced by somewhat early plateauing of error in Fig.~\ref{fig:comp_fnbl}. 

\section{Conclusion}\label{sec:conclude}
To conclude, we have presented a sampling theory approach to inverse scattering
transform which is shown to achieve algebraic orders of convergence provided the
regularity conditions on the input data is fulfilled. The convergence behavior
observed in the numerical experiments with Schwartz class (bandlimited or
effectively bandlimited) impulse response tends to exhibit exponential orders of
convergence. We hope to improve our theoretical estimates to explain these
observations in the future. The complexity of the proposed algorithms depend on
the linear solvers used. A conjugate gradient based iterative solver exhibits a
complexity of $\bigO{N_{\text{iter.}}N^2}$ per sample of the signal computed
where $N$ is the number of sampling basis functions used. Using 
a variant of the classical sampling series due to Helms and Thomas, we were able to 
achieve a complexity of $\bigO{N_{\text{iter.}}N\log N}$ by exploiting the
Hankel symmetry and approximately banded structure of the matrices involved. The
bandedness of the so called quadrature matrix can be controlled by a tolerance
$\epsilon$ which introduces an error of $\bigO{N\epsilon}$ in the computed
solution.

Finally, let us remark that, apart from the avenues of improvement mentioned
above, one can identify several other ways the performance of the
proposed algorithms can be improved. The first one has to do with the nature of 
the basis functions itself. We know from the work of Kaiblinger and
Madych~\cite{KM2006} that orthonormal sampling functions with rapid decay can be 
designed which can potentially reduce the errors committed in arriving at an 
effectively sparse quadrature matrix. Secondly,
the seed for iterative solvers is obtained by using a direct solver at least
once in order to start the algorithm when computing the signal over a grid. In a
parallel implementation this would no longer be a good choice; therefore, our
algorithm can benefit greatly from a cheaper method of ``guessing'' the seed.
 
%%%%%%%%%%%%%%%%%%%%%%%%%%%%%%%%%%%%%%%%%%%%%%%%%%%%%%%%%%%%%%%%%%%%%%%%%%%%%%%%%%%%%%%%%%%

%\bibliography{NBL_PRE}

\providecommand{\noopsort}[1]{}\providecommand{\singleletter}[1]{#1}%
\appendix
\section{The quadrature matrix}\label{sec:q-mat}

The entries of the quadrature matrix, denoted by $\mathcal{Q}$, are defined as
\begin{equation}
\mathcal{Q}_{nl}=\int_{0}^{\infty}\psi_n(s)\psi_l(s)ds,\quad l,n\in\field{Z}.
\end{equation}
It is possible to compute these integrals in terms of the Sine and the Cosine 
integrals which are defined as~\cite[Chap.~6]{Olver:2010:NHMF}
\begin{equation}
\begin{aligned} 
&\Si(t) = \int_0^t\frac{\sin s}{s}ds,
&&\si(t) = -\int_t^{\infty}\frac{\sin s}{s}ds,\\
&\Cin(t)= \int^{t}_{0}\frac{1-\cos s}{s}ds,
&&\Ci(t) = -\int^{\infty}_{t}\frac{\cos s}{s}ds,
\end{aligned}
\end{equation}
and
\begin{equation} 
\begin{split} 
&\si(t) = \Si(t)-\frac{\pi}{2},\\
&\Ci(t) = -\Cin(t)+\log t+\gamma,
\end{split} 
\end{equation}
where $\gamma$ is the Euler's constant. The diagonal entries of $\mathcal{Q}$ works out
to be
\begin{equation}
\mathcal{Q}_{mm}
=\frac{1}{\pi}\int^{\infty}_{0}\frac{\sin^2t}{(t-m\pi)^2}dt
=\frac{1}{2}-\frac{1}{\pi}\Si(-2m\pi).
\end{equation}
Turning to the off-diagonal elements, we have
\begin{equation}
\mathcal{Q}_{mn}=\frac{(-1)^{m+n}}{\pi}\int^{\infty}_{0}\frac{\sin^2t}{(t-m\pi)(t-n\pi)}dt,\quad
m\neq n.
\end{equation}
Note that the integrand of $\mathcal{Q}_{mn}$ is an entire function of $t$. For the moment, assuming
that the origin does not coincide with $n\pi$ or $m\pi$, one can deform the
contour of integration to write
\begin{align*}
&2\pi^2(m-n)(-1)^{m+n}\mathcal{Q}_{mn}\\
&=\lim_{R\rightarrow\infty}\int^{R}_{0}
\left[\frac{1}{(t-m\pi)}-\frac{1}{(t-n\pi)}\right]dt\\
&\qquad-\int^{\infty}_{0}
\left[\frac{1}{(t-m\pi)}-\frac{1}{(t-n\pi)}\right]\cos(2t)dt\\
&=-\log\left(\frac{m}{n}\right)+\Ci(-2m\pi)-\Ci(-2n\pi),
\end{align*}
which yields
\begin{multline}
\mathcal{Q}_{mn}=-\frac{(-1)^{m+n}}{2\pi^2(m-n)}\times\\\left[\Cin(-2m\pi)
-\Cin(-2n\pi)\right],\quad m\neq n.
\end{multline}
Note that the final result does not have any singularities; therefore, we
conclude that it is valid for all $m,n\in\field{Z},\,m\neq n$.
Using the symmetry properties of the Sine and Cosine integrals, we have
\begin{equation}
\mathcal{Q}_{mn}=
\begin{cases}
\frac{1}{2}-\frac{1}{\pi}\Si(2n\pi), & m=n\\
\frac{(-1)^{m+n}}{2\pi^2(m-n)}\left[\Cin(2|m|\pi)-\Cin(2|n|\pi)\right],& m\neq n.
\end{cases}
\end{equation}

\section{Computing the nonlinear impulse response}\label{app:NIR-Bessel}
The input to the inverse NFT is the reflection coefficient $\rho(\xi)$; however,
the GLM equation based approach requires us to compute the nonlinear impulse 
response which is defined by
\begin{equation}\label{eq:nir-int}
p(\tau)=\frac{\sigma}{2\pi}\int_{-1}^{1}\rho(\sigma\xi) e^{i\xi\sigma\tau}d\xi.
\end{equation}
Ordinarily this integral can be computed quite efficiently using the FFT
algorithm which is based on the trapezoidal rule. For large values of the 
quantity $\sigma\tau$, the accuracy of the trapezoidal rule may degrade;
therefore, if extremely high degree of accuracy is demanded we must turn to
other alternatives. It is well known that Gauss-type 
quadrature schemes tend to perform poorly in computing these 
integrals on account of the oscillatory nature of the integrand which
deviates considerably from polynomials, specially for larger values of
$\sigma\tau$. There is a vast amount of literature devoted to treating such
problems, for instance, see~\cite[Section~2.10]{DR1984}) and the references
therein. Here, we would like to choose the method due to Bakhvalov and 
Vasil'eva~\cite{BV1968} which begins with the series expansion
\begin{equation}\label{eq:rho-LT}
\rho(\sigma\xi)=\sum_{n=0}^{\infty}\hat{\rho}_n{\legendre}_n(\xi),\quad\xi\in(-1,1),
\end{equation}
where $\legendre_n(t)$ denotes the Legendre polynomials. Using the
Legendre-Gauss-Lobatto (LGL) nodes, a finite dimensional approximations of
$\hat{\rho}_n$ can be obtained via the Legendre transform~\cite{CHQZ2007}. Let
$\besselj_n(t)$ 
denote the spherical Bessel function of the first kind~\cite[Chap.~10]{Olver:2010:NHMF}.
Now, in order to obtain the exact result, we recall the identity
\begin{equation}
\int^{1}_{-1}{\legendre}_n(\xi)
e^{i\xi \sigma t}d\xi=2i^n\,{\besselj}_n(\sigma t).
\end{equation}
Plugging~\eqref{eq:rho-LT} into~\eqref{eq:nir-int}, we have
\begin{equation}\label{eq:p-besselj}
p(\tau)={\frac{\sigma}{\pi}}
\sum_{n=0}^{\infty}\hat{\rho}_ni^n\,{\besselj}_n(\sigma \tau).
\end{equation}
With precomputed LGL nodes and associated weights, the complexity of obtaining
$\hat{\rho}_n,\,n=0,1,N_{\text{quad.}}-1,$ is $\bigO{N_{\text{quad.}}^2}$ excluding 
the cost of evaluating $\rho(\xi)$. For an efficient method of evaluation of the 
resulting series for $p(\tau)$, one may use the Clenshaw's 
algorithm~\cite{C1955,D1976} which makes efficient use of the recurrence relation 
for the spherical Bessel functions.

\end{document}